\documentclass[12pt]{article}
\usepackage[margin=1in]{geometry}
\usepackage[T1]{fontenc}
\usepackage{lmodern}
\usepackage{amsmath,amssymb,amsthm,mathtools}
\usepackage{booktabs}
\usepackage{longtable}
\usepackage{array}
\usepackage{graphicx}
\graphicspath{{figures/}}
\usepackage{placeins}
\usepackage{float}
\usepackage[font=small,labelfont=bf,labelsep=period]{caption}
\usepackage{enumitem}
\usepackage{microtype}
\usepackage[hidelinks,hypertexnames=false]{hyperref}
\hypersetup{
  pdftitle={Equal-charge projection of the N=4 index},
  pdfauthor={Miguel Tierz},
  pdfsubject={Equal-charge projection, multigraviton index, finite-rank coefficients},
  pdfkeywords={superconformal index, equal-charge projection, AdS5 black holes, giant gravitons}
}
\allowdisplaybreaks[1]
\setlist{topsep=0.35em,itemsep=0.2em,parsep=0.1em,leftmargin=*}
\renewcommand{\arraystretch}{1.06}

\newcommand{\CT}{\operatorname{CT}}
\newcommand{\mc}{\mathcal}

\newcommand{\pfn}{\mathsf{p}}

\theoremstyle{plain}
\newtheorem{theorem}{Theorem}[section]
\newtheorem{proposition}[theorem]{Proposition}
\newtheorem{lemma}[theorem]{Lemma}
\newtheorem{corollary}[theorem]{Corollary}

\theoremstyle{remark}
\newtheorem{remark}[theorem]{Remark}

\newcommand{\arxiv}[1]{\href{https://arxiv.org/abs/#1}{arXiv:#1}}

\title{Equal-charge projection of the $\mc N=4$ index:\\ exact large-$N$ formula and finite-rank $U(3)$ coefficients}
\author{Miguel Tierz\\[4pt]
\textit{Shanghai Institute for Mathematics and}\\
\textit{Interdisciplinary Sciences (SIMIS)}\\
\textit{No.~657 Songhu Road, Yangpu District, Shanghai 200433, China}\\[2pt]
\texttt{tierz@simis.cn}}
\date{}

\begin{document}
\raggedbottom
\maketitle

\begin{abstract}
The equal-charge branch of supersymmetric rotating AdS$_5$ black holes has $Q_1=Q_2=Q_3$.  The corresponding microcanonical sector of the $\mc N=4$ superconformal index is obtained by projecting to equal charges, or equivalently by extracting the constant term in the two charge-difference fugacities. 
We prove that for the large-$N$ multigraviton sector the projected index factorizes exactly as
\[
  \mc I^{\rm eqQ}_\infty(x,p)
  =\prod_{k\ge1}(1-p^kx^{3k})(1-p^{-k}x^{3k})
    \sum_{n\ge0}\pfn(n)^3x^{6n},
\]
where $\pfn(n)$ is the partition function. This factorization gives, for every spin sector, an explicit onset energy below which the large-$N$ coefficient is zero.  Exact $U(3)$ computations show that finite-rank coefficients can nevertheless appear at energies where the large-$N$ coefficient vanishes, including beyond the classical $U(3)$ black-hole bound.  We also determine the full line $j=6J_R+6$.  In particular, with $j^*(J_R)$ denoting this large-$N$ onset energy,
\[
  d_3^{\rm eqQ}(87,\tfrac{27}{2})=1,
  \qquad
  j^*(\tfrac{27}{2})-87=1554,
\]
and the first giant-graviton sector already contributes one unit at this point.  All coefficients are coefficients of the $(-1)^F$-graded index, not positive degeneracies.  The main conclusion is that the high-spin tail survives the exact equal-charge projection.
\end{abstract}

\section{Introduction}
\label{sec:intro}

The superconformal index of $\mc N=4$ $U(N)$ Yang--Mills theory~\cite{KMMR,Romelsberger} is a protected, $(-1)^F$-graded count of BPS states~\cite{Witten-index}.  In the AdS/CFT correspondence~\cite{Maldacena}, suitable large-charge limits of this index reproduce the entropy of supersymmetric AdS$_5$ black holes~\cite{CKMN,CBM,BM,Honda,ACMM,Ardehali-Cardy,BCSZZ}, building on earlier three-dimensional localization results~\cite{BZ}.  Subleading corrections~\cite{GHLPZ}, finite-rank Bethe-Ansatz and residue-sum techniques~\cite{GHLPZ-BA,vLMR}, and the finite-$N$ BPS classification program~\cite{GGKM,Chang-Yin,CL-fortuity,CL-words,CCKLL} extend this comparison beyond leading order.  A related test based on fortuity classification in matrix-model BPS spectra was developed in~\cite{Tierz-BPS}.

Recent exact finite-$N$ calculations~\cite{ACKK,Murthy-GG,Murthy-growth,PZ} show that the index contains a high-spin tail.  In the equal-fugacity notation
\begin{equation}
  \mc I_N(x,p)=\sum_{j,J_R} d_N(j,J_R)\,x^j\,p^{2J_R},
  \label{eq:dN-def}
\end{equation}
the coefficients can be nonzero even when $J_R$ lies beyond the classical black-hole existence bound~\cite{PZ}.  Here $j\in\mathbb Z_{\ge0}$ is the energy grading and $J_R\in\frac12\mathbb Z$ is the right-moving angular momentum.  The coefficients $d_N(j,J_R)$ are protected signed integers, because the trace runs over $\tfrac{1}{16}$-BPS states annihilated by one chosen supercharge and its conjugate~\cite{KMMR}.\footnote{The gradings $(j,2J_R)$ are fixed by their values on the elementary BPS letters, $(2,0)$ for each of the three chiral scalars and $(3,\pm1)$ for the two BPS derivatives, cf.~eq.~\eqref{eq:fmref}; in the conventions of~\cite{KMMR,PZ} this corresponds to $j=2(E+J_L)$, with $E$ the conformal weight and $J_L,J_R$ the left and right angular momenta.}  The fugacities $x$ and $p$ are treated as formal variables, with $|x|<1$ ensuring convergence of the matrix integral.  In this equal-fugacity form the three $R$-charges $Q_1,Q_2,Q_3$ are suppressed.  Possible interpretations of the tail include finite-$N$ wrapped-brane corrections in the giant-graviton expansion~\cite{GE,Imamura,LR,Murthy-GG}, the grey-galaxy phase predicted for sufficiently asymmetric angular momenta in the equal-charge large-$N$ index~\cite{grey1,grey2,CJKKLMP,Bajaj-grey}, and more general partially deconfined phases~\cite{AMMPVR,AArdehali}.

The equal-charge sector is singled out for a physical reason: the equal-charge branch of supersymmetric rotating AdS$_5$ black holes has $Q_1=Q_2=Q_3$.  The dual microcanonical ensemble should therefore fix the two charge differences to zero, rather than average over them.

The distinction between equal fugacities and equal charges is simple but important.  Let $y_1,y_2$ be formal Laurent variables conjugate to $Q_1-Q_3$ and $Q_2-Q_3$.  Setting $y_1=y_2=1$ imposes equal chemical potentials, but it does not restrict the trace to states with $Q_1=Q_2=Q_3$.  In the fully refined index a term
\begin{equation}
  c\,x^j\,p^{2J_R}\,y_1^{\,Q_1-Q_3}\,y_2^{\,Q_2-Q_3}
  \label{eq:refined-term}
\end{equation}
contributes to the equal-fugacity coefficient for every value of $(Q_1-Q_3,Q_2-Q_3)$.  Thus the equal-fugacity index at fixed $(j,J_R)$ still sums over all charge assignments.  The equal-charge projection instead extracts the coefficient of $y_1^0y_2^0$, i.e.\ the sector $Q_1-Q_3=Q_2-Q_3=0$.  This projection fixes the two charge differences, but it does not turn the index into a positive degeneracy at fixed five charges: the common charge and $J_L$ are still summed at fixed $(j,J_R)$.  The exact projection was recognized in the finite-$N$ study of~\cite{PZ}, but left aside as computationally difficult.  Explicitly~\cite{KMMR}, the equal-fugacity index is
\begin{equation}
  \mc I_N(x,p)=\frac{1}{N!}
  \int_{[0,2\pi)^{N-1}}
  |\Delta(\theta)|^2\,
  \exp\!\left[\sum_{m=1}^\infty\frac{f_m(x,p)}{m}\,
  S_m(\theta)\right]
  \frac{d\theta}{(2\pi)^{N-1}},
  \label{eq:IN}
\end{equation}
\noindent with
\begin{align}
  f_m(x,p)&=1-\frac{(1-x^{2m})^3}
  {(1-p^m x^{3m})(1-x^{3m}/p^m)},\notag\\
  |\Delta(\theta)|^2&=\textstyle\prod_{a<b}
  |e^{i\theta_a}-e^{i\theta_b}|^2,\qquad
  S_m(\theta)=N+2\textstyle\sum_{a<b}\cos m(\theta_a-\theta_b).
\end{align}
The equal charges $Q_1=Q_2=Q_3$ correspond to the diagonal $U(1)\subset U(1)^3$.  The general $U(1)^3$ supersymmetric AdS$_5$ black holes have three charges and two angular momenta, subject to one constraint relating them~\cite{KLR}.  On the equal-charge branch there is a single independent $R$-charge and two independent angular momenta; this is the $SO(6)$ gauged-supergravity construction of \cite{CCLP-U13}, equivalently their minimal gauged-supergravity solution~\cite{CCLP}.  The Gutowski--Reall black hole~\cite{GR}, the first supersymmetric AdS$_5$ black hole, is the further specialization with equal angular momenta.  The equal-charge, arbitrary-$J_R$ sector studied here is thus the natural finite-rank analogue of these black holes.  On the equal-charge slice the charge dependence of the Bekenstein--Hawking entropy $S_{\rm BH}=2\pi\sqrt{Q_1Q_2+Q_1Q_3+Q_2Q_3-\tfrac{\pi}{4G_Ng^3}(J_1+J_2)}$~\cite{HHZ,CKMN,Honda} becomes the symmetric combination $3Q^2$, while the dependence on $J_1+J_2$ is retained.
On the gravity side, a charged rotating hairy phase in the same equal-charge, equal-angular-momentum sector was constructed in~\cite{DMS}; in the region where it coexists with the supersymmetric black-hole branch above, this phase dominates the microcanonical ensemble.

\medskip\noindent\textit{Relation to the large-$N$ grey-galaxy analysis.}
The large-$N$ grey-galaxy analysis of~\cite{CJKKLMP} predicts a
microcanonical phase diagram of the index.  On the equal-charge cut
$Q_1=Q_2=Q_3$, the black-hole phase dominates at small $|J_R|$, while
a rank-two grey-galaxy phase dominates at larger $|J_R|$.  The
numerical test in that work expands an equal-potential index in the
variables $A=(Q_1+Q_2+Q_3)/3+J_L$ and $J_R$.  Here $J_L$ is the left-moving angular momentum.  Thus it fixes the
average electric charge but still sums over charge differences.

\paragraph{Main question.}
The present paper asks the complementary finite-rank question: after the exact projection $Q_1=Q_2=Q_3$, does the high-spin tail remain, and how much of it is already present in the large-$N$ multigraviton sector?

We carry out that projection in two steps, first at large $N$ and
then at finite rank.  The large-$N$ step gives the multigraviton index, the
closed formula, Theorem~\ref{thm:main},
\begin{equation}
  \mc I_\infty^{\rm eqQ}(x,p)
  =\Pi(x,p)\sum_{n\ge0}\pfn(n)^3x^{6n},
  \qquad
  \Pi(x,p)=\prod_{k\ge1}(1-p^kx^{3k})(1-p^{-k}x^{3k}).
  \label{eq:intro-main-formula}
\end{equation}
Here $\pfn(n)$ is the number of integer partitions of $n$:
\[
  \sum_{n\ge0}\pfn(n)\,q^n=\prod_{k\ge1}(1-q^k)^{-1}.
\]
The formula has three transparent ingredients:
\begin{enumerate}[label=(\roman*)]
\item the plethystic exponential turns the single-particle tower into the multigraviton product;
\item the equal-charge constant term performs a diagonal extraction among the three $R$-charge partition generators, producing the cube $\pfn(n)^3$ and the grading $x^{6n}$;
\item Euler's pentagonal theorem leaves nonzero terms only at the generalized pentagonal exponents in the prefactor depending on the angular fugacity $p$.
\end{enumerate}
At fixed spin $m=2J_R$, the possible multigraviton contributions are
controlled by pairs of generalized pentagonal numbers whose difference is
$m$.  Proposition~\ref{prop:onset-divisor} turns this into a finite divisor
minimization for $2m$, giving the onset energy and the onset sign directly
from divisor data.
The intervals used below are therefore exact consequences of the product formula: in them the multigraviton coefficient is exactly zero, not merely numerically small.

The finite-rank step uses this exact large-$N$ result as a test.  The coefficients of
$\Pi$ are supported on generalized pentagonal numbers~\cite{Andrews},
\[
  \mc P=\bigl\{\,\tfrac{k(3k-1)}{2}:k\in\mathbb Z\,\bigr\}
  =\{0,1,2,5,7,12,15,22,26,\ldots\}.
\]
For fixed $J_R$ define the onset energy $j^*(J_R)$ to be the first energy at
which the large-$N$ coefficient can be nonzero.  The interval
\[
  6J_R\le j<j^*(J_R)
\]
is then an interval where the large-$N$ coefficient vanishes: by construction,
$d_\infty^{\rm eqQ}(j,J_R)=0$ throughout it.  A nonzero finite-rank
coefficient in this interval is therefore not supplied by the multigraviton index.
Throughout, ``multigraviton'' refers strictly to the large-$N$ index
$\mc I_\infty$, the plethystic exponential of the single-particle index.
At finite rank, write
\[
  G_N^{\rm eqQ}:=d_N^{\rm eqQ}-d_\infty^{\rm eqQ}.
\]
There are three relevant possibilities.
\begin{itemize}
\item \textbf{No finite-rank correction.}  If $d_N^{\rm eqQ}=d_\infty^{\rm eqQ}$, then $G_N^{\rm eqQ}=0$.
\item \textbf{Absent at large $N$, present at finite rank.}  If $d_\infty^{\rm eqQ}=0$ but $d_N^{\rm eqQ}\neq0$, the coefficient is absent from the multigraviton index and comes from finite rank.
\item \textbf{Finite-rank modification.}  If $d_\infty^{\rm eqQ}\neq0$ and $d_N^{\rm eqQ}\neq d_\infty^{\rm eqQ}$, finite-rank corrections modify the multigraviton value, possibly by partial cancellation, exact cancellation, sign reversal, or an increase in magnitude.
\end{itemize}
We will highlight the cases that also lie beyond the classical $U(3)$ black-hole
existence bound $J_R^{\rm BH}(j)$, defined in eq.~\eqref{eq:JRBH}.  These are precisely the coefficients satisfying
\begin{equation}
  J_R>J_R^{\rm BH}(j),\qquad
  d_\infty^{\rm eqQ}(j,J_R)=0,
  \qquad
  d_3^{\rm eqQ}(j,J_R)\ne0.
\end{equation}
For example, at $j=108$ the beyond-boundary spin sectors are
\begin{equation}
\begin{array}{c|ccccc}
  j=108,\quad J_R & 14 & 15 & 16 & 17 & 18\\
\hline
  d_\infty^{\rm eqQ} & 0 & 0 & 0 & 1 & 0\\
  d_3^{\rm eqQ}      & -5648 & -162 & -1 & 1 & 0
\end{array}.
\label{eq:strip108}
\end{equation}
Thus three adjacent sectors are absent from the multigraviton answer but
present at finite rank.  The next sector, $J_R=17$, is the first large-$N$
onset, while $J_R=18=j/6$ lies on the kinematic boundary, where
Proposition~\ref{prop:frontier} gives $G_N^{\rm eqQ}=0$ for all $N$.
Table~\ref{tab:onset} gives $j^*(27/2)=1641$, so the example found here with
the largest depth below onset is
\[
  d_3^{\rm eqQ}(87,27/2)=1,
  \qquad j^*(27/2)-87=1641-87=1554.
\]
Thus the high-spin tail is not created merely by summing over unequal $R$-charge assignments.

\paragraph{Scope of the claims.}
We separate analytic results, exact finite-rank computations, and interpretive evidence.  Analytically, we prove the large-$N$ factorization and its charge-resolved refinement, the onset formula for $j^*(J_R)$, the intervals where the large-$N$ coefficient vanishes, the rank-independence of the kinematic boundary, and the formula on the first adjacent line $j=6J_R+6$.  In particular, the zero intervals used below follow from the
pentagonal/divisor structure of the large-$N$ product, rather than from a
finite truncation of the series.  Computationally, the displayed $U(3)$ coefficients are obtained exactly in arithmetic over integers, not floating-point arithmetic; these include the complete $j=66$ slice and the selected high-spin examples.  Because the refined $U(3)$ integrand factorizes as $U(1)\times SU(3)$ when the common center acts trivially on adjoint letters, Table~\ref{tab:su3-adjoint-check} also reports genuine $SU(3)$ coefficients at selected targets, although not a complete $SU(3)$ scan.

At the sector level, a direct computation gives $h_{3,1}^{\rm eqQ}(87,\tfrac{27}{2})=1$, where $h_{N,M}^{\rm eqQ}$ is the projected contribution of the $M$-th giant-graviton sector defined in eq.~\eqref{eq:h-def}.  The remaining sector checks and the exact aggregate relation are in Section~\ref{sec:GG-decomposition}; they are not inputs to the main support claim.  Broader wrapped-brane interpretations and near-boundary growth patterns should be read as observations suggested by the data, not as theorems.

Three caveats should remain visible throughout:
\begin{enumerate}[label=(\alph*)]
\item the coefficients are $(-1)^F$-graded signed coefficients of the index, not positive BPS degeneracies;
\item the main finite-rank data set is a $U(3)$ data set, while the $SU(3)$ entries are selected checks of the interacting sector;
\item the aggregate correction $G_3^{\rm eqQ}$ is known exactly, but only selected $M$-giant sectors are evaluated individually.
\end{enumerate}

\paragraph{Organization.}
Section~\ref{sec:theorem} defines the projection, proves the equal-charge and charge-resolved large-$N$ factorization formulae, and records the formulas on $j=6J_R$ and $j=6J_R+6$.  Section~\ref{sec:zero-intervals} uses the factorized large-$N$ formula to determine the onset energy $j^*(J_R)$ and the resulting intervals where the large-$N$ coefficient vanishes.  Section~\ref{sec:eqQcomputation} gives the finite-rank computation, the main $U(3)$ examples, and the selected $SU(3)$ check.  Section~\ref{sec:BH-physics} records the black-hole entropy comparison and the partial giant-graviton sector check.  Section~\ref{sec:discussion} summarizes the interpretation, limitations, and remaining checks.

\section{Equal-charge projection of the multigraviton index}
\label{sec:theorem}

This section sets up the projection and the large-$N$ benchmark used throughout the paper.  We prove the projected large-$N$ factorization (Theorem~\ref{thm:main}), record the charge-resolved refinement, and give the finite-$N$ formulas on the two boundary lines $j=6J_R$ and $j=6J_R+6$.

\paragraph{Convention.}
We work with the $U(N)$ index, retaining the free $U(1)$ center.  This is the natural setting for the matrix-integral formula~\eqref{eq:IN} and the giant-graviton expansion~\cite{GE,Imamura,LR}, and it matches the conventions of~\cite{PZ,ACKK}.  In AdS$_5$/CFT$_4$, the interacting sector is $SU(N)$, and at large~$N$ the free $U(1)$ decouples from the leading entropy.  At $N=3$ the distinction can matter at subleading order.  The black-hole existence bound~\eqref{eq:JRBH} used here is the $U(3)$ bound from~\cite{PZ}~eq.~(4.7), not the $SU(3)$ bound; all tail comparisons are therefore internally consistent within the $U(N)$ convention.  Genuine $SU(3)$ coefficients for selected examples where the large-$N$ coefficient vanishes are reported in Section~\ref{sec:U1-center}.

The logic of the computation is summarized schematically in Figure~\ref{fig:schematic}.

\begin{figure}[t]
\centering
\includegraphics[width=\textwidth]{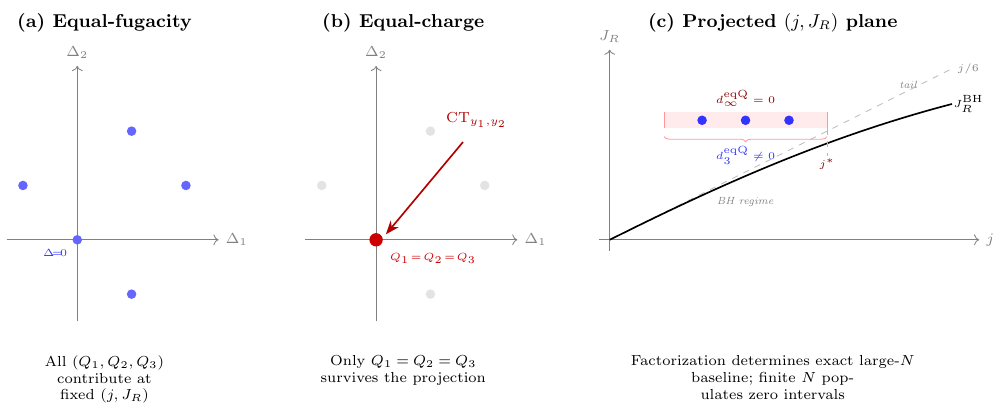}
\caption{Schematic of the equal-charge projection.  The equal-fugacity index sums over the full $A_2$ charge-difference lattice at fixed $(j,J_R)$, whereas the equal-charge projection keeps only the origin $Q_1=Q_2=Q_3$.  In the projected $(j,J_R)$ plane, the large-$N$ factorization gives exact intervals where the large-$N$ coefficient is zero; finite-rank coefficients can still be nonzero there, even beyond the black-hole bound.}
\label{fig:schematic}
\end{figure}

\subsection{The index and its equal-charge projection}
\label{sec:setup}

With all four fugacities retained, the refined index is~\cite{PZ}
\begin{equation}
  \mc I_N(x,p,y_1,y_2)
  =\operatorname{Tr}_{\mc H_{\rm BPS}}(-1)^F\,x^j\,p^{2J_R}\,
  y_1^{Q_1-Q_3}\,y_2^{Q_2-Q_3},
  \label{eq:Iref}
\end{equation}
\noindent where the trace runs over the Hilbert space of gauge-invariant BPS operators.  The refined index is computed by the same matrix integral~\eqref{eq:IN}, with $f_m$ replaced by the refined single-particle index $f_m^{\rm ref}$, evaluated at the $m$-th power of each fugacity:
\begin{equation}
  f_m^{\rm ref}(x,p,y_1,y_2)
  =1-\frac{(1-y_1^m x^{2m})(1-y_2^m x^{2m})
  (1-x^{2m}/(y_1 y_2)^m)}
  {(1-p^m x^{3m})(1-x^{3m}/p^m)}.
  \label{eq:fmref}
\end{equation}
Setting $y_1=y_2=1$ performs an equal-potential specialization, not an equal-charge projection.  A typical term in the refined index has the form
\[
  c\,x^j\,p^{2J_R}\,y_1^{\Delta_1}y_2^{\Delta_2},
  \qquad
  \Delta_i=Q_i-Q_3,
\]
where $c$ is a numerical coefficient.  When $y_1=y_2=1$, this contributes $c\,x^j p^{2J_R}$ regardless of the charge differences $(\Delta_1,\Delta_2)$.  Thus unequal-charge sectors are still summed when $y_1=y_2=1$.  The microcanonical equal-charge condition keeps only the terms with $\Delta_1=\Delta_2=0$:
\begin{equation}
  \mc I_N^{\rm eqQ}(x,p)
  :=\CT_{y_1,y_2}\,
  \mc I_N(x,p,y_1,y_2),
  \label{eq:IeqQ}
\end{equation}
\noindent where $\CT_{y_1,y_2}$ denotes the constant-term extraction in $(y_1,y_2)$: the coefficient of $y_1^0\,y_2^0$ in the Laurent expansion.

Each monomial in $f_m^{\rm ref}$ arises from expanding the denominator as a geometric series in powers $r,s\ge0$ and selecting one of the four numerator factors, labeled by $\ell\in\{0,1,2,3\}$.  The resulting charges are
\[
  j=2m\ell+3m(r+s),\qquad 2J_R=m(r-s).
\]
Since $j-6|J_R|=2m\ell+6m\min(r,s)\ge0$:
\begin{equation}
  |J_R|\le\frac{j}{6}.
  \label{eq:kinematic}
\end{equation}
This bound follows directly from the charge relations above.  The bound $J_R\le j/6$ is called kinematic because it follows from the letter charges before any dynamics, gauge integration, or choice of rank enters.  Equivalently, it follows from $j=2m\ell+3m(r+s)$ and $2J_R=m(r-s)$ alone.  We record it explicitly because, after the equal-charge projection, the index becomes $N$-independent on $J_R=j/6$ (Proposition~\ref{prop:frontier}).

The large-$N$ limit of the index factorizes into single-particle contributions~\cite{KMMR}.  In the giant graviton expansion~\cite{GE,Imamura,LR,Lee-GGE}, this multigraviton sector is the zeroth-order term, and finite-$N$ corrections are organized by wrapped-brane contributions.  The giant graviton itself was introduced in~\cite{MST,GMT,HHI}.  The multigraviton index is the plethystic exponential~\cite{BFHH} of the single-particle index.  For a formal Laurent series $g$ with zero constant term,
\[
  \mathrm{PE}[g]
  :=
  \exp\!\left(
  \sum_{m\ge1}\frac{1}{m}
  g(x^m,p^m,y_1^m,y_2^m)
  \right),
\]
so in particular $\mathrm{PE}[q/(1{-}q)]=\prod_{k\ge1}(1-q^k)^{-1}$ and $\mathrm{PE}[-q/(1{-}q)]=\prod_{k\ge1}(1-q^k)$.  Writing $a=y_1 x^2$, $b=y_2 x^2$, $c=x^2/(y_1 y_2)$, $d=px^3$, $e=p^{-1}x^3$, the large-$N$ single-particle index is
\[
  i_{\rm sp}
  =\frac{a}{1{-}a}+\frac{b}{1{-}b}+\frac{c}{1{-}c}
  -\frac{d}{1{-}d}-\frac{e}{1{-}e}.
\]
Applying $\mathrm{PE}$ term by term gives
\begin{equation}
  \mc I_\infty(x,p,y_1,y_2)
  =\prod_{k=1}^\infty
  \frac{(1-p^k x^{3k})(1-x^{3k}/p^k)}
  {(1-y_1^k x^{2k})(1-y_2^k x^{2k})
  (1-x^{2k}/(y_1 y_2)^k)}.
  \label{eq:Iinf_full}
\end{equation}
At equal fugacity ($y_1=y_2=1$):
\begin{equation}
  \mc I_\infty(x,p)=\prod_{k\ge1}
  \frac{(1-p^k x^{3k})(1-x^{3k}/p^k)}{(1-x^{2k})^3}.
  \label{eq:Iinf_eqfug}
\end{equation}

\subsection{The factorization theorem}

This subsection states the large-$N$ factorization and explains why it matters for the rest of the paper.  The formula separates the equal-charge multigraviton index into a prefactor depending on the angular fugacity $p$ and a partition series.  In Section~\ref{sec:zero-intervals}, the fact that this prefactor is nonzero only at pentagonal exponents will give explicit onset energies and intervals where the large-$N$ coefficient vanishes.

\begin{theorem}[Equal-charge multigraviton factorization]
\label{thm:main}
\begin{equation}
  \mc I_\infty^{\rm eqQ}(x,p)
  =\Pi(x,p)\cdot\sum_{n=0}^\infty \pfn(n)^3\,x^{6n},
  \label{eq:factorization}
\end{equation}
\noindent where $\pfn(n)$ is the integer partition function and
\begin{equation}
  \Pi(x,p):=\prod_{k=1}^\infty
  (1-p^k x^{3k})(1-p^{-k}x^{3k}).
  \label{eq:Pi}
\end{equation}
\end{theorem}

\noindent
The role of this formula is not only to compute $\mc I_\infty^{\rm eqQ}(x,p)$: it gives an exact large-$N$ multigraviton index against which the finite-$N$ calculation is compared.  For each $J_R$, the pentagonal onset $j^*(J_R)$ bounds an interval of energies in which no multigraviton contribution can occur.  Any nonzero finite-$N$ coefficient in that interval is therefore a finite-rank contribution not present in the large-$N$ index.  Figure~\ref{fig:beyond-BH} plots such coefficients using two coordinates defined in Section~\ref{sec:zero-intervals}.

\begin{proof}
The numerator~$\Pi(x,p)$ of~\eqref{eq:Iinf_full}
depends only on $x$ and $p$, since the charge-difference fugacities $y_1,y_2$ appear only in the denominator.  It therefore factors out of
the constant-term extraction:
\begin{equation}
  \mc I_\infty^{\rm eqQ}(x,p)
  =\Pi(x,p)\cdot\CT_{y_1,y_2}\!
  \left[\prod_{k\ge1}
  \frac{1}{(1-y_1^k x^{2k})
  (1-y_2^k x^{2k})
  (1-(y_1 y_2)^{-k}x^{2k})}\right].
  \label{eq:step1}
\end{equation}
Each denominator factor is a partition generating function.  Writing
\[
  \mathrm{Part}(q):=\prod_{k\ge1}(1-q^k)^{-1}=\sum_{n\ge0}\pfn(n)\,q^n,
\]
the constant-term factor in~\eqref{eq:step1} is
\[
  \CT_{y_1,y_2}\,
  \mathrm{Part}(y_1 x^2)\,
  \mathrm{Part}(y_2 x^2)\,
  \mathrm{Part}\!\left(\frac{x^2}{y_1 y_2}\right).
\]
Expanding the three partition generators and extracting the coefficient of $y_1^0 y_2^0$ gives the diagonal identity:
\begin{equation}
\label{eq:diagonal}
\begin{aligned}
&\CT_{y_1,y_2}
\prod_{k\ge1}
  \frac{1}{(1-y_1^k x^{2k})(1-y_2^k x^{2k})
  (1-(y_1 y_2)^{-k}x^{2k})}
\\
&=
\CT_{y_1,y_2}
\sum_{a,b,c\ge0}
\pfn(a)\pfn(b)\pfn(c)\,
x^{2(a+b+c)}\,
y_1^{a-c}\,y_2^{b-c}
\\
&=
\sum_{n\ge0}\pfn(n)^3\,x^{6n},
\end{aligned}
\end{equation}
where the last step uses $a-c=b-c=0$, i.e.\ $a=b=c=:n$.
\end{proof}

\noindent
All infinite-product manipulations in the proof are justified coefficientwise in the formal $x$-adic expansion; the identity~\eqref{eq:factorization} holds in the ring of formal power series $\mathbb{Z}[p,p^{-1}][\![x]\!]$.

\medskip\noindent\textit{Origin of the partition-function cube.}
The cube arises because there are three $R$-charge
fugacities $(y_1,y_2,(y_1 y_2)^{-1})$ subject to
two constant-term constraints, leaving one diagonal
index: the three partition sizes associated with the three charge directions must be equal.  This is the origin of both the cube $\pfn(n)^3$ and the energy grading $x^{6n}$.  For a denominator of the same product-of-geometric-series structure
with $k$ independent charge fugacities and $k{-}1$ constraints,
the analogous formula would give $\pfn(n)^k$.

\begin{remark}[Symmetric form]
\label{rem:qqbar}
In the variables $q=px^3$, $\bar q=x^3/p$ used on the kinematic boundary below,
Theorem~\ref{thm:main} reads
\[
  \mc I^{\rm eqQ}_\infty
  =(q;q)_\infty\,(\bar q;\bar q)_\infty\,\Phi(q\bar q),
  \qquad
  \Phi(t):=\sum_{n\ge0}\pfn(n)^3\,t^n,
\]
with $(t;t)_\infty:=\prod_{k\ge1}(1-t^k)$, making the $p\leftrightarrow p^{-1}$
symmetry of the index manifest.
\end{remark}

The diagonal extraction in the proof of Theorem~\ref{thm:main} works equally
well with the constant term replaced by an arbitrary charge sector, by the same calculation; this gives the exact large-$N$ multigraviton index for the full
charge-difference profile, not only for its equal-charge slice.

\begin{proposition}[Charge-resolved factorization]
\label{prop:charge-resolved}
For $\Delta_1,\Delta_2\in\mathbb Z$,
\begin{equation}
  [y_1^{\Delta_1}y_2^{\Delta_2}]\;\mc I_\infty(x,p,y_1,y_2)
  \;=\;\Pi(x,p)\;x^{2(\Delta_1+\Delta_2)}
  \sum_{c\,\ge\,\max(0,-\Delta_1,-\Delta_2)}
  \pfn(c)\,\pfn(c{+}\Delta_1)\,\pfn(c{+}\Delta_2)\;x^{6c}.
  \label{eq:charge-resolved}
\end{equation}
Theorem~\ref{thm:main} is the case $\Delta_1=\Delta_2=0$.
\end{proposition}

\begin{proof}
Identical to the proof of Theorem~\ref{thm:main}: in the diagonal
identity~\eqref{eq:diagonal}, extract the coefficient of
$y_1^{\Delta_1}y_2^{\Delta_2}$ instead of the constant term.  The conditions
$a-c=\Delta_1$, $b-c=\Delta_2$ give $a=c+\Delta_1$, $b=c+\Delta_2$ with
$c\ge\max(0,-\Delta_1,-\Delta_2)$, and
$x^{2(a+b+c)}=x^{6c+2(\Delta_1+\Delta_2)}$.
\end{proof}

\noindent
Two structural facts are immediate from~\eqref{eq:charge-resolved} and are
used later.  First, the right-hand side depends on $(\Delta_1,\Delta_2)$ only
through the multiset $\{0,\Delta_1,\Delta_2\}$ up to a common shift: the
substitution $c\to c-1$ identifies, for example, the
$(\Delta_1,\Delta_2)=(1,-1)$ and $(2,1)$ profiles.  This is the $S_3$ permutation symmetry of the three $R$-charge directions $(y_1,y_2,(y_1y_2)^{-1})$, not a symmetry of the gauge-group roots.  The same charge-permutation symmetry holds
at finite $N$, where it acts on the three refined letter directions in
$f_m^{\rm ref}$; in
Appendix~\ref{app:computational} the equality of the $(1,-1)$ and
$(2,1)$ charge-resolved coefficients is therefore used as a consistency check.

Second, every
monomial of~\eqref{eq:charge-resolved} obeys
$j-6J_R-2(\Delta_1+\Delta_2)\in6\mathbb Z$, which is the support congruence
established for all finite $N$ in Appendix~\ref{app:computational}.  At fixed
$(j,J_R)$ the allowed $(\Delta_1,\Delta_2)$ therefore form a coset of the
index-three sublattice $\{\Delta_1+\Delta_2\in3\mathbb Z\}$ of the
charge-difference lattice; writing the shifted multiset as
$(-\delta,\Delta_1{-}\delta,\Delta_2{-}\delta)$ with
$\delta=(\Delta_1+\Delta_2)/3$ identifies this sublattice with the triangular
$A_2$ root lattice $\{(a_1,a_2,a_3)\in\mathbb Z^3:a_1{+}a_2{+}a_3=0\}$.  This
is the $A_2$ lattice of Figure~\ref{fig:schematic} and the support sublattice
of Section~\ref{sec:BH-physics}.  A
quantitative consequence for the equal-charge suppression factor is derived in
Section~\ref{sec:BH-physics}.

\begin{corollary}[The $p=1$ specialization]
\label{cor:p1}
At $p=1$, the prefactor specializes to $\Pi(x,1)=\prod_{k\ge1}(1-x^{3k})^2$, so the projected total
\begin{equation}
  D_\infty(j)
  :=[x^j]\,\mc I_\infty^{\rm eqQ}(x,1)
  =[x^j]\left[
  \prod_{k\ge1}(1-x^{3k})^2
  \cdot\sum_{n\ge0}\pfn(n)^3\,x^{6n}\right].
  \label{eq:Dinf}
\end{equation}
Setting $p=1$ gives $D_\infty(j)=\sum_{J_R}d_\infty^{\rm eqQ}(j,J_R)$, the sum running over all kinematically allowed~$J_R$; by the $p\leftrightarrow p^{-1}$ symmetry of the index~\cite{KMMR} this equals $d_\infty^{\rm eqQ}(j,0)+2\sum_{J_R>0}d_\infty^{\rm eqQ}(j,J_R)$.
\end{corollary}

\begin{proposition}[Kinematic-boundary stability]
\label{prop:frontier}
Let $q=x^3p$.  Define the generalized pentagonal numbers
\[
  \omega(r):=\frac{r(3r-1)}{2},\qquad r\in\mathbb Z,
\]
and define $E(n)$ by Euler's pentagonal theorem~\cite{Andrews},
\[
  (q;q)_\infty
  :=\prod_{k\ge1}(1-q^k)
  =\sum_{r\in\mathbb Z}(-1)^r\,q^{\omega(r)}
  =:\sum_{n\ge0}E(n)\,q^n.
\]
Equivalently, $E(n)=0$ unless $n=\omega(r)$ is a generalized pentagonal number, in which case $E(n)=(-1)^r$.
On the positive kinematic boundary $j=6J_R$, $J_R\ge0$, the
equal-charge-projected finite-$N$ index is independent of $N$:
\begin{equation}
  \sum_{n\ge0} d_N^{\rm eqQ}(3n,n/2)\,q^n
  =
  (q;q)_\infty.
  \label{eq:frontier}
\end{equation}
In particular, for
integer $J_R=J$,
\begin{equation}
  d_N^{\rm eqQ}(6J,J)=E(2J),
  \qquad
  G_N^{\rm eqQ}(6J,J)=0.
  \label{eq:frontier-integer}
\end{equation}
Thus finite-$N$ giant-graviton corrections do not reach the kinematic boundary.  The negative boundary $J_R\le0$ follows by $p\leftrightarrow p^{-1}$.
\end{proposition}

\begin{proof}
We work on the positive-spin boundary $J_R\ge0$; the opposite boundary is related by $p\leftrightarrow p^{-1}$.  Write $P=2J_R$.
For a single monomial in the $m$-th single-letter contribution, obtained by choosing denominator powers $r,s\ge0$ and a numerator degree $\ell\in\{0,1,2,3\}$, one has
\[
  j=2m\ell+3m(r+s),\qquad P=m(r-s).
\]
Hence the boundary defect of each monomial is
\[
  j-3P=2m\ell+6ms\ge0.
\]
This defect is additive under multiplication of letters.  Therefore a product contributing to a boundary coefficient with $j=3P$, equivalently $j=6J_R$, can contain only monomials with zero defect.  Zero defect forces $\ell=0$ and $s=0$: all $x^{2m}$-numerator letters and all $p^{-m}x^{3m}$-denominator letters are excluded on the boundary.  In particular every surviving letter is neutral in $(y_1,y_2)$, so the equal-charge projection acts trivially on the boundary and the computation below applies verbatim to the projected index $\mc I_N^{\rm eqQ}$.  The surviving single-letter boundary part is
\begin{equation}
  f_m^\partial(q)
  =
  1-\frac{1}{1-q^m}
  =
  -\frac{q^m}{1-q^m}.
\end{equation}
Writing $z_i$ for the $U(N)$ eigenvalues, the kinematic boundary contribution is
\begin{equation}
  \mc I_N^\partial(q)
  =
  \frac{1}{N!}\CT_{\mathbf z}
  \prod_{i\ne j}(1-z_i/z_j)
  \prod_{r\ge1}\prod_{i,j}(1-q^r z_i/z_j).
\end{equation}
Using
\[
  \prod_{r\ge1}\prod_{i,j}(1-q^r z_i/z_j)
  =(q;q)_\infty^N
  \prod_{i\ne j}(qz_i/z_j;q)_\infty,
\]
we obtain
\begin{equation}
  \mc I_N^\partial(q)
  =
  (q;q)_\infty^N\,
  \frac{1}{N!}\CT_{\mathbf z}
  \prod_{i\ne j}(z_i/z_j;q)_\infty.
\end{equation}
The $q$-Dyson/Macdonald constant-term identity~\cite{Macdonald,ZB} gives
\begin{equation}
  \frac{1}{N!}\CT_{\mathbf z}
  \prod_{i\ne j}(z_i/z_j;q)_\infty
  =
  (q;q)_\infty^{1-N}.
\end{equation}
For background on the Selberg/Morris constant-term circle and its
random-matrix connections, see~\cite{FW-Selberg,Forrester-book}.
Therefore $\mc I_N^\partial(q)=(q;q)_\infty$.
\end{proof}

The boundary computation shows that no finite-rank correction reaches the
kinematic ceiling $j=6J_R$.  The support congruence
$j-6J_R\in6\mathbb Z$ then makes the next possible line
\[
  j=6J_R+6.
\]
This first adjacent line is a useful test of how far the boundary stability
persists into the interior.  The answer is almost completely stable: for
$N\ge4$ the coefficient already equals the large-$N$ multigraviton value.
Rank three is the first exceptional case, and the exception is controlled by
the first non-stable cubic constant-term identity.  This gives a closed formula for
several of the high-spin examples used later, including the deeply delayed
coefficient $(j,J_R)=(87,\tfrac{27}{2})$.

\begin{proposition}[The line $j=6J_R+6$]
\label{prop:strip}
Let $P=2J_R\in\mathbb Z_{\ge0}$, and set $E(n)=0$ for $n<0$. On the line
\[
j=3P+6=6J_R+6
\]
one has
\[
d^{\rm eqQ}_3(3P+6,P/2)
=
E(P)-E(P+1)+E(P-1)-E(P-2),
\]
whereas, for every $N\ge4$,
\[
d^{\rm eqQ}_N(3P+6,P/2)
=
E(P)-E(P+1)
=
d^{\rm eqQ}_\infty(3P+6,P/2).
\]
Equivalently,
\[
G^{\rm eqQ}_3(3P+6,P/2)
=
E(P-1)-E(P-2),
\qquad
G^{\rm eqQ}_N(3P+6,P/2)=0
\quad (N\ge4).
\]
The negative-spin line follows by the symmetry $p\leftrightarrow p^{-1}$.
\end{proposition}

\begin{proof}
We expand away from the positive kinematic boundary.  Set
\[
  q=px^3,\qquad h=x^2,\qquad P=2J_R.
\]
Then
\[
  x^j p^P=q^P h^{(j-3P)/2}.
\]
Thus the boundary line $j=3P=6J_R$ is the coefficient of $h^0$, while the
first adjacent line $j=3P+6=6J_R+6$ is the coefficient of $h^3$.

For the $m$-th single-letter contribution write
\[
  S_m=\sum_{a,b=1}^N (z_a/z_b)^m,
\]
and
\[
  Y_m=y_1^m+y_2^m+(y_1y_2)^{-m},
  \qquad
  Y_m^\vee=y_1^my_2^m+y_1^{-m}+y_2^{-m}.
\]
Expanding the refined single-letter index in powers of $h$ gives
\[
  f_m^{\rm ref}
  =
  -\frac{q^m}{1-q^m}
  +
  \frac{h^mY_m-h^{2m}Y_m^\vee}{1-q^m}
  -
  h^{3m}q^{-m}
  +O(h^{4m}).
\]
The first term is exactly the boundary term used in
Proposition~\ref{prop:frontier}; the remaining terms measure the defect away
from the boundary.

Factoring off the boundary contribution turns the matrix integral into the
same normalized constant-term expectation as before.  Namely, with
\[
  Z_N(q)=
  \frac1{N!}\CT_{\mathbf z}
  \prod_{i\ne j}(z_i/z_j;q)_\infty
  =
  (q;q)_\infty^{1-N}
\]
and
\[
  \langle F\rangle_N
  =
  Z_N(q)^{-1}
  \frac1{N!}\CT_{\mathbf z}
  F
  \prod_{i\ne j}(z_i/z_j;q)_\infty ,
\]
the coefficient on the line $j=3P+6$ is obtained by multiplying
$(q;q)_\infty$ by the normalized expectation of the
$h^3y_1^0y_2^0$ defect insertion.

It remains to find this insertion.  We need the $h^3y_1^0y_2^0$ part of
\[
  \exp\left[
  \sum_{m\ge1}\frac{1}{m}
  \left(
  \frac{h^mY_m-h^{2m}Y_m^\vee}{1-q^m}
  -
  h^{3m}q^{-m}
  +O(h^{4m})
  \right)S_m
  \right].
\]
At order $h^3$, the $m=3$ term contains $Y_3$, whose constant term in
$y_1,y_2$ is zero.  The mixed $m=1,m=2$ term contains $Y_1Y_2$, whose
constant term is also zero.  Hence the only neutral contributions are
\[
  -h^3q^{-1}S_1,
\]
\[
  \left(\frac{hY_1S_1}{1-q}\right)
  \left(-\frac{h^2Y_1^\vee S_1}{1-q}\right),
\]
and
\[
  \frac{1}{6}
  \left(\frac{hY_1S_1}{1-q}\right)^3.
\]
Since
\[
  \CT_{y_1,y_2}(Y_1Y_1^\vee)=3,
  \qquad
  \CT_{y_1,y_2}(Y_1^3)=6,
\]
the equal-charge extraction leaves
\[
  \mathcal O_1
  =
  -q^{-1}S_1
  -
  \frac{3S_1^2}{(1-q)^2}
  +
  \frac{S_1^3}{(1-q)^3}.
\]

We now apply Lemma~\ref{lem:inserted-constant-terms}.  In the ranks needed here,
\[
  \langle S_1\rangle_N=1-q,
  \qquad
  \langle S_1^2\rangle_N=2(1-q)^2,
\]
and
\[
  \langle S_1^3\rangle_3=(1-q)^3(6+q-q^2),
  \qquad
  \langle S_1^3\rangle_N=6(1-q)^3\quad(N\ge4).
\]
Therefore
\[
  \langle\mathcal O_1\rangle_3
  =
  -q^{-1}(1-q)-6+(6+q-q^2)
  =
  1-q^{-1}+q-q^2,
\]
whereas
\[
  \langle\mathcal O_1\rangle_N
  =
  -q^{-1}(1-q)-6+6
  =
  1-q^{-1}
  \qquad(N\ge4).
\]

Finally we restore the boundary factor.  Since
\[
  (q;q)_\infty=\sum_{n\ge0}E(n)q^n,
\]
we get
\[
  d^{\rm eqQ}_3(3P+6,P/2)
  =
  [q^P](q;q)_\infty(1-q^{-1}+q-q^2),
\]
and
\[
  d^{\rm eqQ}_N(3P+6,P/2)
  =
  [q^P](q;q)_\infty(1-q^{-1})
  \qquad(N\ge4).
\]
Using the convention $E(n)=0$ for $n<0$, these become
\[
  d^{\rm eqQ}_3(3P+6,P/2)
  =
  E(P)-E(P+1)+E(P-1)-E(P-2),
\]
and
\[
  d^{\rm eqQ}_N(3P+6,P/2)
  =
  E(P)-E(P+1)
  \qquad(N\ge4).
\]
The last expression is precisely the large-$N$ value on this line, obtained
from eq.~\eqref{eq:formula} with $b=0$ and $b=1$.  This proves the stated
formulae.  The negative-spin line follows from the $p\leftrightarrow p^{-1}$
symmetry.
\end{proof}

Thus the first line above the kinematic boundary is still stable for
all $N\ge4$.  The rank-three deviation is entirely caused by the
first non-stable cubic constant-term identity and is encoded by the simple
Euler-pentagonal correction $E(P-1)-E(P-2)$.

\medskip\noindent\textit{Kinematic-boundary coefficient.}
The ceiling coefficient $d_N^{\rm eqQ}(j,j/6)$ is nonzero precisely when $j/3=2J_R$ is a generalized pentagonal number; the controlling index is $j/3$, not $j/6$.  For instance, $j=90$ has $j/6=15\in\mc P$ but $j/3=30\notin\mc P$, so $d_N^{\rm eqQ}(90,15)=E(30)=0$.

Writing $j=6M$ and $J_R=M-r$ with $r\ge0$, the multigraviton coefficient near the kinematic boundary becomes
\begin{equation}
  d_\infty^{\rm eqQ}(6M,M-r)
  =\sum_{0\le b\le r}
  E(2M{-}2r{+}b)\,E(b)\,\pfn(r{-}b)^3,
  \label{eq:near-frontier}
\end{equation}
since the condition $b\le r$ follows from $\pfn(r{-}b)^3=0$ for $b>r$.
In particular,
\begin{align}
  r=0:&\quad E(2M),\notag\\
  r=1:&\quad E(2M{-}2)-E(2M{-}1),\notag\\
  r=2:&\quad 8\,E(2M{-}4)-E(2M{-}3)-E(2M{-}2).
  \notag
\end{align}
The prefactor $8$ in the $r=2$ line is $\pfn(2)^3=2^3$.

\subsection{Euler pentagonal expansion}
\label{sec:pentagonal}

The pentagonal expansion of the prefactor $\Pi(x,p)$ is central to the onset analysis.  Writing $\Pi=A\cdot B$ with
\begin{align}
  A&=\prod_k(1-p^k x^{3k})=\sum_a E(a)\,p^a\,x^{3a},\notag\\
  B&=\prod_k(1-p^{-k}x^{3k})=\sum_b E(b)\,p^{-b}\,x^{3b}.
\end{align}
Multiplying the two Euler expansions gives
\[
  \Pi(x,p)
  =\sum_{a,b\ge0} E(a)E(b)\,p^{a-b}\,x^{3(a+b)}.
\]
The coefficient of $p^m$ is obtained by imposing $a-b=m$, i.e.\ $a=m+b$:
\begin{equation}
  [p^m]\,\Pi(x,p)
  =\sum_{b=0}^\infty E(m{+}b)\,E(b)\,x^{3m+6b},
  \label{eq:pmN}
\end{equation}
where $E(n)$ is the Euler pentagonal function defined in
Proposition~\ref{prop:frontier}, nonzero only at
generalized pentagonal numbers
$\{0,1,2,5,7,12,15,22,26,\ldots\}$
with values $\pm1$.
Substituting into Theorem~\ref{thm:main} and collecting powers of $p$ gives the explicit formula
\begin{equation}
  d_\infty^{\rm eqQ}(j,J_R)
  =\sum_{b=0}^\infty E(2J_R{+}b)\,E(b)\,
  \pfn\!\left(\frac{j-6J_R-6b}{6}\right)^{\!3},
  \label{eq:formula}
\end{equation}
\noindent where $\pfn(n)=0$ for $n<0$ or $n\notin\mathbb Z$.

More generally, substituting the pentagonal expansion into
Proposition~\ref{prop:charge-resolved} gives the charge-resolved profile
\begin{equation}
  d_\infty^{(\Delta_1,\Delta_2)}(j,J_R)
  :=[x^j p^{2J_R} y_1^{\Delta_1} y_2^{\Delta_2}]\,\mc I_\infty
  =\sum_{b=0}^\infty E(2J_R{+}b)\,E(b)\,
  \pfn(c)\,\pfn(c{+}\Delta_1)\,\pfn(c{+}\Delta_2),
  \label{eq:formula-resolved}
\end{equation}
with $c=\bigl(j-6J_R-6b-2(\Delta_1+\Delta_2)\bigr)/6$ and the convention that
the summand vanishes unless $c\in\mathbb Z$ with
$c\ge\max(0,-\Delta_1,-\Delta_2)$.  The equal-charge
formula~\eqref{eq:formula} is the diagonal $\Delta_1=\Delta_2=0$.

For fixed $m=2J_R$, we call
\[
  (b,a)=(b,m+b)
\]
an \emph{active pentagonal pair} if both $b$ and $a=m+b$ lie in
$\mc P$.  Equivalently,
\[
  b\in B_m:=\{b\in\mc P:m+b\in\mc P\}.
\]
Such a pair can contribute in the spin sector $J_R=m/2$; at a
given energy $j$ it contributes only if
$j\ge 3m+6b$ and $j\equiv3m\pmod6$.

Thus the sum~\eqref{eq:formula} is organized by the active set $B_m$:
all other values of $b$ vanish because at least one of $E(b)$ or
$E(m+b)$ is zero.
For many of the spin sectors used below, the range $j\le90$ lies in a single-pair regime where only one pentagonal pair contributes (cf.~Table~\ref{tab:onset}).  An instructive example is $J_R=4$:
$[p^8]\Pi(x,p)=0$ for $j<66$: the first pair of
generalized pentagonal numbers differing by~$8$ is $(b,a)
=(7,15)$, giving $E(15)\,E(7)=(-1)(+1)=-1$ at
$j=3(a+b)=3(15+7)=66$.  So
\begin{equation}
  d_\infty^{\rm eqQ}(j,4)=0
  \qquad\text{for all }j<66,
  \label{eq:JR4vanish}
\end{equation}
and the entire $J_R=4$ content at equal charges
below $j=66$ comes from finite-rank corrections.

\medskip\noindent
\textit{Support congruence.}
Formula~\eqref{eq:formula} shows that a multigraviton coefficient can be nonzero only if
\begin{equation}
  j-6J_R \in 6\mathbb Z_{\ge0}.
  \label{eq:support_congruence}
\end{equation}
Equivalently, integer $J_R$ contributes only at $j\equiv0\pmod{6}$,
while half-integer $J_R$ contributes only at $j\equiv3\pmod{6}$.
The same congruence holds at every finite~$N$ (Appendix~\ref{app:computational}).

The explicit formula~\eqref{eq:formula} can be reorganized around the finite
set of active pentagonal pairs, revealing a cancellation structure controlled
by signed contributions.  In the single-pair regime, the asymptotic growth
satisfies
$\log|d_\infty^{\rm eqQ}(j,J_R)|\sim\pi\sqrt{j-j^*(J_R)}$
as $j\to\infty$ at fixed~$J_R$, by the Hardy--Ramanujan
formula~\cite{HR}.  The fixed-$J_R$ asymptotic discussion is given at the end
of Section~\ref{sec:zero-intervals}; the onset-family derivations are
collected in Appendix~\ref{app:onset-proofs}.
\section{\texorpdfstring{Intervals with zero large-$N$ coefficient}{Intervals with zero large-N coefficient}}
\label{sec:zero-intervals}

The factorization theorem identifies, for each $J_R$, an exact energy threshold below which the large-$N$ coefficient vanishes.  This section derives these thresholds and the resulting intervals where the large-$N$ coefficient vanishes.

\noindent\textit{Onset thresholds.}
For each $J_R\ge0$, the first nonzero multigraviton coefficient
occurs at
\begin{equation}
  j^*(J_R):=6J_R+6b_{\min}(J_R),
  \qquad
  b_{\min}(J_R):=\min\{b\in\mc P : 2J_R+b\in\mc P\},
  \label{eq:onset}
\end{equation}
\noindent where $\mc P=\{0,1,2,5,7,12,15,22,26,35,40,\ldots\}$ is the set
of generalized pentagonal numbers~\cite{Andrews}.  For $j<j^*(J_R)$, the multigraviton sector is absent; any nonzero finite-$N$ coefficient in this interval comes from finite-rank corrections.  To see this, note that from eq.~\eqref{eq:formula}, $d_\infty^{\rm eqQ}(j,J_R)\ne0$ requires $b\in\mc P$ with $2J_R+b\in\mc P$ and $j\ge6J_R+6b$; the minimum is attained at $b=b_{\min}(J_R)$.

\begin{remark}[Onset algorithm]
Given $J_R$, set $m=2J_R$.  Find the smallest $b\in\mc P$ such that $m+b\in\mc P$.  Then $j^*(J_R)=6J_R+6b$.
\end{remark}

For fixed $J_R$, the set of energies $\{j:6J_R\le j<j^*(J_R)\}$ has no multigraviton contribution: $d_\infty^{\rm eqQ}(j,J_R)=0$ there because the prefactor $\Pi$ has no term there.  Any nonzero $d_N^{\rm eqQ}(j,J_R)$ in this range is therefore a finite-rank contribution.  In the giant-graviton expansion, it contributes to the aggregate wrapped-brane correction.

When such an interval overlaps the region beyond the black-hole bound $J_R>J_R^{\rm BH}(j)$, a nonzero finite-$N$ coefficient there is simultaneously beyond the BH bound and absent from the large-$N$ index.  It is convenient to measure these two features by
\begin{equation}
  \Delta_{\rm BH}(j,J_R):=J_R-J_R^{\rm BH}(j),
  \qquad
  \Delta_\infty(j,J_R):=j^*(J_R)-j,
  \label{eq:deltas}
\end{equation}
so that the coefficients of interest are those with $\Delta_{\rm BH}>0$, $\Delta_\infty>0$, and $d_3^{\rm eqQ}(j,J_R)\neq0$.

\begin{figure}[htbp]
\centering
\includegraphics[width=0.65\textwidth]{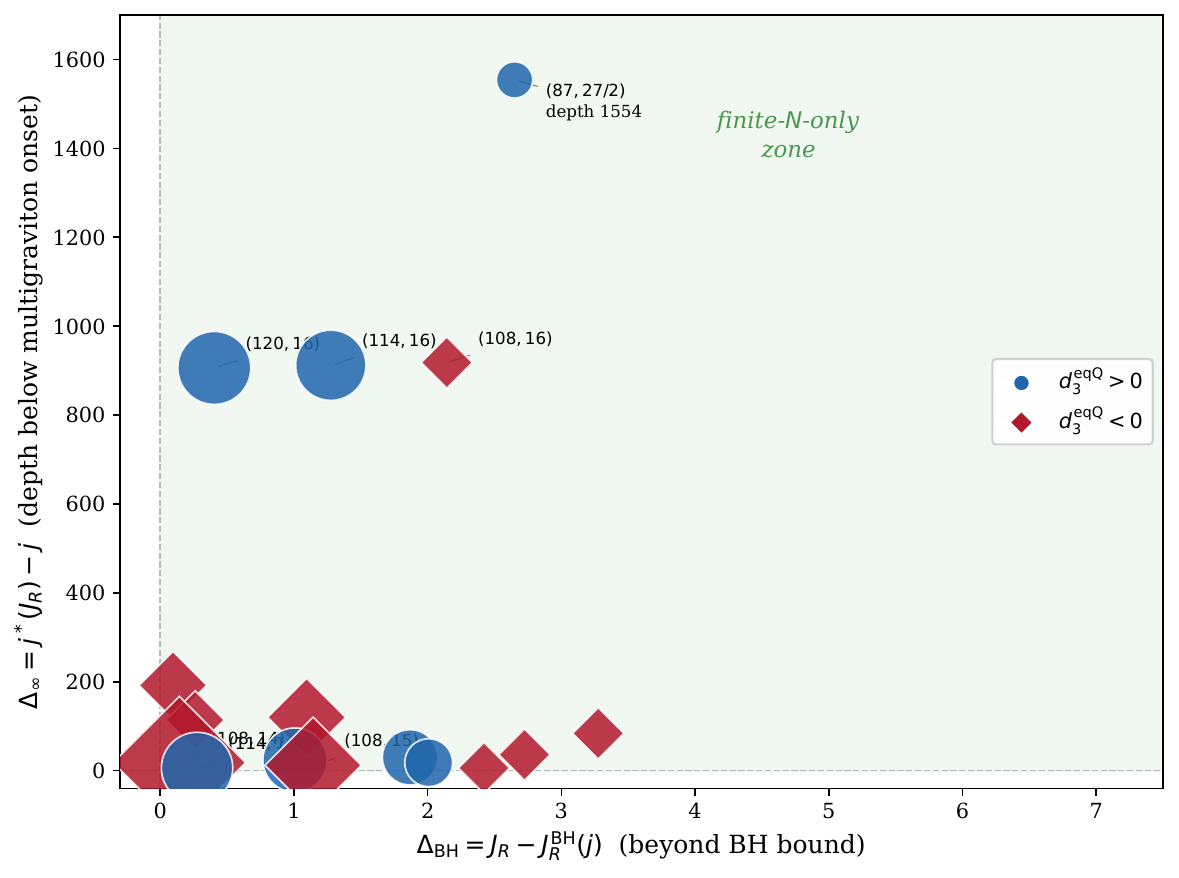}
\caption{Coefficients beyond the black-hole bound and below the large-$N$ onset.  The horizontal axis is the distance beyond the BH bound, $\Delta_{\rm BH}=J_R-J_R^{\rm BH}(j)$; the vertical axis is the depth below the large-$N$ onset, $\Delta_\infty=j^*(J_R)-j$.  Points in the upper-right quadrant are beyond the BH bound and absent from the multigraviton sector.  Marker size is proportional to $\log(1+|d_3^{\rm eqQ}|)$; circles are positive coefficients and diamonds are negative.}
\label{fig:beyond-BH}
\end{figure}
Such a sector is beyond the classical BH bound and has no multigraviton support.  In the giant-graviton expansion~\cite{GE,Imamura,LR}, these coefficients are naturally interpreted as aggregate wrapped-brane corrections.  Except for the partial sector check at $(87,27/2)$ in Section~\ref{sec:GG-decomposition}, the individual $M$-giant decomposition is not checked in this paper.

The opposite phenomenon also occurs: the finite-$N$ correction can exactly erase a multigraviton tail coefficient, so that $d_3^{\rm eqQ}(j,J_R)=0$ even though $d_\infty^{\rm eqQ}(j,J_R)\ne0$ (an example is $(72,10)$ below).

The onset data needed in the rest of the paper are summarized in Table~\ref{tab:onset}.  The large thresholds, such as $j^*(8)=258$, $j^*(16)=1026$, and $j^*(27/2)=1641$, are the source of the long intervals in which the large-$N$ coefficient is zero.  The divisor criterion is stated as Proposition~\ref{prop:onset-divisor} below; its proof, the power-of-two family, the $m=2^\beta3^\alpha$ delayed-onset family, and the checks above the onset are collected in Appendix~\ref{app:onset-proofs}.  The main text uses the resulting threshold values and the implication
\[
  j<j^*(J_R) \quad\Longrightarrow\quad d_\infty^{\rm eqQ}(j,J_R)=0.
\]

The onset problem is equivalent to finding the first pair of generalized
pentagonal numbers whose difference is $m=2J_R$.  Writing the pair as
$(a,b)=(\omega(r),\omega(s))$, the difference factors:
\[
  \omega(r)-\omega(s)
  =
  \frac{(r-s)(3(r+s)-1)}{2}.
\]
Thus every active pentagonal pair determines a factorization
$uv=2m$, with $u=r-s$ and $v=3(r+s)-1$.  The congruence and parity
conditions below are exactly the conditions that the inverse formulae for
$r$ and $s$ give integers.  The onset is therefore a finite divisor
minimization rather than a search through pentagonal numbers.

\begin{proposition}[Onset thresholds from divisor minimization]
\label{prop:onset-divisor}
Let $J_R>0$ and $m=2J_R\in\mathbb Z_{>0}$.  Call a factorization $uv=2m$ with
$u,v\in\mathbb Z$ \emph{admissible} if
\[
  v\equiv-1\pmod3,
  \qquad
  u\equiv\frac{v+1}{3}\pmod2.
\]
Admissible factorizations are in bijection with the active pentagonal pairs
$b\in B_m:=\{b\in\mc P:\,m+b\in\mc P\}$ via $u=r-s$, $v=3(r+s)-1$,
$(a,b)=(\omega(r),\omega(s))$, and the candidate onset energy of a pair is
\begin{equation}
  j\;=\;3(a+b)\;=\;\frac{9u^2+v^2-1}{4}.
  \label{eq:onset-energy}
\end{equation}
Distinct admissible factorizations give distinct candidate energies.
Consequently
\begin{equation}
  j^*(J_R)
  \;=\;\min\left\{\frac{9u^2+v^2-1}{4}\;:\;uv=2m\ \text{admissible}\right\},
  \label{eq:onset-min}
\end{equation}
the minimizer $(u_*,v_*)$ is unique, and the onset coefficient is
\begin{equation}
  d_\infty^{\rm eqQ}\bigl(j^*(J_R),J_R\bigr)
  \;=\;(-1)^{(v_*+1)/3}\in\{\pm1\}.
  \label{eq:onset-sign}
\end{equation}
The set of admissible factorizations is nonempty---it always contains
$(u,v)=(-2m,-1)$, with candidate energy $9m^2$---and $|B_m|$ equals its
cardinality.
\end{proposition}

\noindent
The proof is in Appendix~\ref{app:onset-proofs}.  Both delayed-onset families
quoted below, eqs.~\eqref{eq:pow2} and~\eqref{eq:twoadic}, follow by
enumerating the admissible factorizations of $2m=2^{n+1}$ and
$2m=2^{\beta+1}3^\alpha$, and Table~\ref{tab:onset} is the output
of~\eqref{eq:onset-min} for $2J_R\le40$.  Combining the proposition with the
post-onset support analysis of Appendix~\ref{app:onset-proofs} gives the
following support statement.  The vanishing-interval implication is a theorem for
all $J_R$; the converse is a theorem at every grade (i.e.\ every fixed pair
$(j,J_R)$) where at most three pentagonal pairs contribute, and otherwise a
finite-range computational check.

For fixed $m=2J_R$, define
\[
  B_m:=\{b\in\mc P:m+b\in\mc P\},
  \qquad
  \sigma_b:=E(m+b)\,E(b)\in\{\pm1\}
  \quad (b\in B_m).
\]
When the support congruence $j-6J_R\in6\mathbb Z_{\ge0}$ holds, write
\[
  t:=\frac{j-6J_R}{6}=\frac{j-3m}{6}\in\mathbb Z_{\ge0}.
\]
Then eq.~\eqref{eq:formula} becomes
\begin{equation}
  d_\infty^{\rm eqQ}(j,J_R)
  =
  D_m(t)
  :=
  \sum_{\substack{b\in B_m\\ b\le t}}
  \sigma_b\,\pfn(t-b)^3.
  \label{eq:Dmt}
\end{equation}

\begin{theorem}[Multigraviton support and checked converse]
\label{thm:support}
Fix $J_R\ge0$ and write $m=2J_R$.  A multigraviton coefficient can be nonzero
only on the congruence class $j\equiv6J_R\pmod6$.  On that class,
\begin{equation}
  j<j^*(J_R)
  \quad\Longrightarrow\quad
  d_\infty^{\rm eqQ}(j,J_R)=0,
  \label{eq:zero-before-onset}
\end{equation}
and the onset coefficient is always
\begin{equation}
  d_\infty^{\rm eqQ}\bigl(j^*(J_R),J_R\bigr)=\pm1.
  \label{eq:onset-nonzero}
\end{equation}
For $m>0$, list $B_m=\{b_1<b_2<\cdots\}$ in increasing order and define the
\emph{fourth-pair onset}
\[
  j_4(m):=3m+6b_4\ \ \text{if }|B_m|\ge4,
  \qquad
  j_4(m):=\infty\ \ \text{otherwise},
\]
so that at most three pentagonal pairs contribute to~\eqref{eq:formula}
whenever $j<j_4(m)$.  If $J_R=0$, or if $m>0$ and $j<j_4(m)$, then the full
support on this congruence class is
\begin{equation}
  d_\infty^{\rm eqQ}(j,J_R)=0
  \quad\Longleftrightarrow\quad
  j<j^*(J_R)
  \ \ \text{or}\ \
  (j,J_R)=\bigl(9,\tfrac12\bigr).
  \label{eq:exact-support}
\end{equation}
In particular the equivalence holds at every energy whenever $|B_m|\le3$.
The smallest fourth-pair onset in the entire $(j,J_R)$ plane is $j_4(5)=225$,
at $J_R=\tfrac52$; for grades with four or more contributing pairs, the
same equivalence has been verified computationally for all allowed grades with $j\le9000$.
\end{theorem}

\begin{proof}
From eq.~\eqref{eq:formula}, a nonzero coefficient can occur only if
\[
  j-6J_R\in6\mathbb Z_{\ge0},
\]
which is the support congruence~\eqref{eq:support_congruence}.  Assume from
now on that this congruence holds, set $m=2J_R$, and write
\[
  t=\frac{j-6J_R}{6}=\frac{j-3m}{6}.
\]
Then $d_\infty^{\rm eqQ}(j,J_R)=D_m(t)$ as in~\eqref{eq:Dmt}.

Let $b_1=\min B_m$.  If $t<b_1$, the sum defining $D_m(t)$ is empty, so
$D_m(t)=0$.  Since
\[
  j^*(J_R)=6J_R+6b_1=3m+6b_1,
\]
this proves~\eqref{eq:zero-before-onset}.

At $t=b_1$, only the first pentagonal pair contributes.  Hence
\[
  D_m(b_1)
  =
  \sigma_{b_1}\,\pfn(0)^3
  =
  \sigma_{b_1}
  =
  \pm1,
\]
which is~\eqref{eq:onset-nonzero}.  For $m>0$ the sign is identified by
Proposition~\ref{prop:onset-divisor}; for $m=0$ it is simply $+1$.

If $J_R=0$, then $m=0$ and every active term has sign
\[
  \sigma_b=E(b)^2=1.
\]
Moreover the $b=0$ term contributes $\pfn(t)^3>0$ for every $t\ge0$.
Therefore $D_0(t)>0$ for every allowed grade $t\ge0$.

Now assume $m>0$.  List the elements of $B_m$ as
\[
  b_1<b_2<\cdots,
\]
and set $b_4=\infty$ if $|B_m|\le3$.  If $t<b_4$, then at most three
pentagonal pairs contribute to $D_m(t)$.

With one contributing pair, $D_m(t)$ is a single nonzero signed cube.

With two contributing pairs, cancellation requires
\[
  \pfn(t-b_1)^3=\pfn(t-b_2)^3.
\]
Since $b_1<b_2$, one has $t-b_1>t-b_2\ge0$.  The partition function is
strictly increasing for $n\ge1$~\cite{Andrews}, and the only equality at the
bottom is $\pfn(0)=\pfn(1)=1$.  Hence cancellation forces
\[
  t-b_1=1,\qquad t-b_2=0,
\]
so $b_2-b_1=1$.

The only consecutive integers that are both generalized pentagonal are
$(0,1)$ and $(1,2)$.  If $(b_1,b_2)=(0,1)$, the active-pair condition requires
$m$ and $m+1$ to be consecutive generalized pentagonal numbers.  Since
$m>0$, this forces $m=1$, and the two signs are opposite.  Thus
$t=1$, so
\[
  (m,t)=(1,1),
  \qquad
  (j,J_R)=\left(3m+6t,\frac m2\right)
  =
  \left(9,\frac12\right).
\]
If $(b_1,b_2)=(1,2)$, the active-pair condition forces $m=0$, which is the
$J_R=0$ case already handled above; in that case the signs are the same and
there is no cancellation.  Therefore the unique two-pair cancellation is
$(j,J_R)=(9,\tfrac12)$.

With three contributing pairs, a cancellation can occur only if the signs split
one term against two terms.  After moving terms to opposite sides, this would
give
\[
  X^3=Y^3+Z^3
\]
for positive integers $X,Y,Z$, namely partition numbers.  This is impossible
by Euler's $n=3$ case of Fermat's theorem~\cite{HW}.  The bottom collision
$\pfn(0)=\pfn(1)=1$ creates no exception: it only gives special cases such as
$1^3+1^3=X^3$, which are also impossible.

Thus, for $m>0$ and $t<b_4$, equivalently $j<j_4(m)$, the coefficient can
vanish only below onset or at the exceptional point
$(j,J_R)=(9,\tfrac12)$.  Together with the $J_R=0$ case above, this proves
\eqref{eq:exact-support} in the stated analytic range.

It remains to record the finite-range check when four or more pentagonal pairs
contribute.  The exhaustive computation in
Appendix~\ref{app:onset-proofs} shows that the smallest fourth-pair onset is
\[
  j_4(5)=225,
  \qquad J_R=\frac52,
\]
and that all $108{,}102$ grades with four or more contributing pairs and
$j\le9000$ are nonzero.  This proves the checked converse through $j=9000$.
The grade list and the code reproducing this scan are included as ancillary
files with the arXiv submission.
\end{proof}

\noindent
In words: on its congruence class, the multigraviton coefficient vanishes
strictly below the onset, equals $\pm1$ exactly at the onset, and is nonzero
at every energy above it, with the single exception
$(j,J_R)=(9,\tfrac12)$---unconditionally wherever at most three pentagonal
pairs are in play, and verified through $j=9000$ elsewhere.

\medskip\noindent\textit{Consecutive-index benchmark.}
Write $m=2J_R$ and recall $\omega(k)=k(3k-1)/2$.  A pentagonal pair
$(b,a)=(\omega(s),\omega(r))$ with $a-b=m$ gives the onset candidate
$j=3(a+b)$.  A particularly simple subfamily occurs when the labels $r,s$ are
consecutive.  Here ``consecutive'' refers to the labels in $\omega(k)$, not to
the values $\omega(r)$ and $\omega(s)$ being consecutive integers.

If $|r-s|=1$, then the divisor variable $u=r-s$ satisfies $|u|=1$.  Since
$uv=2m$, one has $v=\pm2m$, and the candidate energy in
Proposition~\ref{prop:onset-divisor} becomes
\[
  j_{\rm cons}(J_R)
  =
  \frac{9+4m^2-1}{4}
  =
  m^2+2
  =
  4J_R^2+2.
\]
The admissibility conditions are satisfied for this $|u|=1$ factorization
exactly when $m\not\equiv0\pmod3$.  Thus, whenever $m\not\equiv0\pmod3$,
there is a consecutive-index candidate and hence
\[
  j^*(J_R)\le j_{\rm cons}(J_R).
\]
This is only a benchmark: another admissible divisor factorization of $2m$ may
give a smaller onset.  When no smaller admissible factorization exists, the
consecutive-index candidate is the actual onset.  In particular, the
power-of-two family in Appendix~\ref{app:onset-proofs} proves this situation
for $m=2^n$, giving the entries $j^*(4)=66$, $j^*(8)=258$, and
$j^*(16)=1026$ in Table~\ref{tab:onset}.

\begin{table}[!htbp]
\caption{Onset thresholds $j^*(J_R)$ and their pentagonal data.
Bold entries mark delayed onsets.
For $2J_R\not\equiv0\pmod{3}$, a consecutive-index pentagonal candidate exists with benchmark $j_{\rm cons}=4J_R^2+2$; entries such as $J_R=4,8,16$ are cases where this candidate is the actual first onset.
For $2J_R\equiv0\pmod{3}$, no consecutive-index candidate exists, and delayed onsets arise from the next admissible divisor factorization.
Derivations are collected in Appendix~\ref{app:onset-proofs}.}
\label{tab:onset}
\begin{center}
\footnotesize
\renewcommand{\arraystretch}{0.88}
\setlength{\tabcolsep}{4pt}
\begin{tabular}{r|rr|rr|r|r}
\toprule
$J_R$ & $b_{\min}$ & $a=2J_R{+}b$ & $E(a)$ & $E(b)$ & $j^*(J_R)$ & onset coeff.\ \\
\midrule
$0$ & $0$ & $0$ & $+1$ & $+1$ & $0$ & $+1$ \\
$1$ & $0$ & $2$ & $-1$ & $+1$ & $6$ & $-1$ \\
$2$ & $1$ & $5$ & $+1$ & $-1$ & $18$ & $-1$ \\
$3$ & $1$ & $7$ & $+1$ & $-1$ & $24$ & $-1$ \\
$4$ & $7$ & $15$ & $-1$ & $+1$ & $\mathbf{66}$ & $-1$ \\
$5$ & $2$ & $12$ & $-1$ & $-1$ & $42$ & $+1$ \\
$6$ & $0$ & $12$ & $-1$ & $+1$ & $36$ & $-1$ \\
$7$ & $1$ & $15$ & $-1$ & $-1$ & $48$ & $+1$ \\
$8$ & $35$ & $51$ & $+1$ & $-1$ & $\mathbf{258}$ & $-1$ \\
$\tfrac{9}{2}$ & $26$ & $35$ & $-1$ & $+1$ & $\mathbf{183}$ & $-1$ \\
$9$ & $22$ & $40$ & $-1$ & $+1$ & $\mathbf{186}$ & $-1$ \\
$10$ & $2$ & $22$ & $+1$ & $-1$ & $72$ & $-1$ \\
$11$ & $0$ & $22$ & $+1$ & $+1$ & $66$ & $+1$ \\
$12$ & $2$ & $26$ & $+1$ & $-1$ & $84$ & $-1$ \\
$13$ & $0$ & $26$ & $+1$ & $+1$ & $78$ & $+1$ \\
$14$ & $7$ & $35$ & $-1$ & $+1$ & $126$ & $-1$ \\
$15$ & $5$ & $35$ & $-1$ & $+1$ & $120$ & $-1$ \\
$16$ & $155$ & $187$ & $-1$ & $+1$ & $\mathbf{1026}$ & $-1$ \\
$17$ & $1$ & $35$ & $-1$ & $-1$ & $108$ & $+1$ \\
$\tfrac{27}{2}$ & $260$ & $287$ & $+1$ & $-1$ & $\mathbf{1641}$ & $-1$ \\
$18$ & $15$ & $51$ & $+1$ & $-1$ & $\mathbf{198}$ & $-1$ \\
$19$ & $2$ & $40$ & $-1$ & $-1$ & $126$ & $+1$ \\
$20$ & $0$ & $40$ & $-1$ & $+1$ & $120$ & $-1$ \\
\bottomrule
\end{tabular}
\end{center}
\end{table}

Table~\ref{tab:onset} also gives the delayed threshold $j^*(27/2)=1641$, the source of the $\Delta_\infty=1554$ large-$N$ vanishing interval around the coefficient $(87,27/2)$ used in Sections~\ref{sec:targeted} and~\ref{sec:GG-decomposition}.

All formulas for $d_\infty^{\rm eqQ}$ have been verified through $j=120$ by an independent constant-term extraction on the refined denominator product, with zero mismatches at every grade.  The support analysis has been checked for all $(j,J_R)$ with $j\le2400$, and at every grade with four or more contributing pentagonal pairs through $j=9000$ (Appendix~\ref{app:onset-proofs}).

\paragraph{Asymptotic growth.}
The explicit formula~\eqref{eq:formula} can be reorganized around the finite set of active pentagonal pairs $B_m:=\{b\in\mc P:m+b\in\mc P\}$ with $m=2J_R$.  For $m>0$, the set $B_m$ is finite: by Proposition~\ref{prop:onset-divisor} it is in bijection with the admissible factorizations of $2m$, and hence has at most as many elements as $2m$ has divisors.  In the single-pair regime ($|B_m|=1$), there is no cancellation and $d_\infty^{\rm eqQ}=\pm \pfn(n)^3$ is determined by one pentagonal pair, with asymptotic growth rate
\begin{equation}
  \log|d_\infty^{\rm eqQ}(j,J_R)|
  \sim \pi\sqrt{j-j^*(J_R)}
  \qquad(j\to\infty,\;J_R\;\text{fixed}),
  \label{eq:asymptotics}
\end{equation}
from the Hardy--Ramanujan formula~\cite{HR} (see also~\cite{Andrews} 
for the general partition asymptotics). The growth analysis 
of the full index is in~\cite{Murthy-growth}.

Indeed, in the single-pair regime $d_\infty^{\rm eqQ}(j,J_R)=\pm\pfn(n)^3$ with $n=(j-j^*(J_R))/6$, so
\[
  \log|\pfn(n)^3|
  \sim 3\pi\sqrt{\frac{2n}{3}}
  =\pi\sqrt{6n}
  =\pi\sqrt{j-j^*(J_R)}.
\]
When $|B_m|\ge2$, the shifted pairs share the same exponential growth rate.  
Indeed, distinct pairs have arguments $n_i=(j-6J_R-6b_i)/6$ differing by 
fixed integers, so $\sqrt{n_i}\sim\sqrt{n_1}$ as $j\to\infty$; a finite 
signed sum of terms with equal leading exponential rate can change the 
prefactor but not the rate.  Partial cancellations among signed contributions 
can therefore suppress the multigraviton coefficient below the na\"ive 
estimate at any given~$j$, but the leading exponential 
rate~\eqref{eq:asymptotics} is unchanged.

\section{\texorpdfstring{Exact $U(3)$ equal-charge computation}{Exact U(3) equal-charge computation}}
\label{sec:eqQcomputation}

This section states exactly what is computed at finite rank.  All numbers are projected $U(3)$ coefficients of the $(-1)^F$-graded index; the large-$N$ multigraviton index is used only as the reference against which finite-rank support is tested.  We then show that the high-spin tail survives the equal-charge projection and present the complete $j=66$ spin-resolved slice.

For the charge-resolved finite-rank coefficients we write
\[
  c_3(j,J_R;\Delta_1,\Delta_2)
  :=[x^j\,p^{2J_R}\,y_1^{\Delta_1}y_2^{\Delta_2}]\,
  \mc I_3(x,p,y_1,y_2),
\]
the finite-rank analogue of the charge-resolved profile $d_\infty^{(\Delta_1,\Delta_2)}$ of eq.~\eqref{eq:formula-resolved}.  Thus
\[
  d_3^{\rm eqQ}(j,J_R)=c_3(j,J_R;0,0).
\]

\paragraph{Computational method.}
All finite-rank $U(3)$ coefficients reported below are obtained by exact arithmetic over integers, not by floating-point numerical approximation.  The pruning condition~\eqref{eq:pruning} removes monomials that cannot contribute to the requested coefficient, but does not approximate the coefficient itself.  High-energy entries are reconstructed by modular arithmetic; the primes, modular reconstruction, and cross-checks are collected in
Appendix~\ref{app:computational}.

We now turn to the finite-$N$ computation, extending the equal-fugacity approach of~\cite{PZ,ACKK} to the exact equal-charge projection.  The $U(3)$ holonomy is a triple of phases $(z_1,z_2,z_3)$.  We change variables to the two ratios $u=z_1/z_2$, $v=z_1/z_3$, and the common phase.  Since all letters are in the adjoint, the integrand is neutral under the common phase and depends only on $u$ and $v$.  The integral over the common phase is therefore trivial, but the observable is still the $U(3)$ index.

The difference between $U(3)$ and $SU(3)$ is carried by the zero weights of the adjoint character: $K(1)^3$ for $U(3)$ and $K(1)^2$ for $SU(3)$ (Section~\ref{sec:U1-center}).  In the variables $u,v$ the adjoint character is
\[
  \chi_{\rm adj}^{U(3)}=3+u+u^{-1}+v+v^{-1}+u/v+v/u.
\]
The three zero weights give the factor $K(1)^3$, the six roots give the six nontrivial $K$-factors, the Vandermonde determinant gives $W(u,v)$, and the Weyl group order gives $1/6$.  The $U(3)$ index with all fugacities is
\begin{equation}
  \mc I_3(x,p,y_1,y_2)
  =\frac{K(1)^3}{6}\,
  \CT_{u,v}\!\Big[
  W(u,v)\,
  K(u)K(u^{-1})K(v)K(v^{-1})K(u/v)K(v/u)
  \Big],
  \label{eq:I3full}
\end{equation}
\noindent where
\begin{align}
  K(z)&=\exp\!\left[\sum_{m\ge1}
  \frac{f_m^{\rm ref}(x,p,y_1,y_2)\,z^m}{m}\right]
  =\sum_{n\ge0}a_n\,z^n,\notag\\
  W(u,v)&=(1-u)(1-u^{-1})(1-v)(1-v^{-1})
  (1-u/v)(1-v/u),
\end{align}
and $a_n$ starts at order $x^{2n}$.

The equal-charge projection is
\begin{equation}
  \mc I_3^{\rm eqQ}(x,p)
  =\CT_{y_1,y_2}\,\mc I_3(x,p,y_1,y_2).
  \label{eq:I3eqQ}
\end{equation}

Because $a_n=O(x^{2n})$, the coefficient of
$x^j$ requires only $a_0,\ldots,a_{\lfloor j/2\rfloor}$.
We compute coefficients of the formal series exactly over the integers, using a compiled routine that multiplies Laurent polynomials stored by their nonzero monomials.
The coefficients are integers; selected high-energy entries whose 64-bit intermediate arithmetic overflows are reconstructed from modular arithmetic, as described in Appendix~\ref{app:computational}.
For selected high-$j$ extractions, we use the following bound to prune monomials that cannot contribute to the target coefficient.

\begin{lemma}[Pruning bound]
\label{lem:pruning}
Let a partial monomial in the expansion of~\eqref{eq:I3full} have energy~$e$ and charge $q=(q_p,q_1,q_2)$ in the $(p,y_1,y_2)$ directions.  For the target coefficient $[x^j\,p^{2J_R}\,y_1^0\,y_2^0]$, define the residual charge
\begin{equation}
  \delta=(\delta_p,\delta_1,\delta_2):=(2J_R-q_p,\,-q_1,\,-q_2).
  \label{eq:residual}
\end{equation}
If the lower bound on the completion energy,
\begin{equation}
  j_{\rm lb}(\delta):=3|\delta_p|+2\bigl(\delta_1+\delta_2+3\max(-\delta_1,-\delta_2,0)\bigr),
  \label{eq:pruning}
\end{equation}
exceeds the remaining energy budget, $j_{\rm lb}(\delta)>j-e$,
then no completion of the partial monomial can contribute to the target coefficient, and the monomial is safely discarded.
\end{lemma}

\noindent The proof, which shows that the energy required to complete the $y$-charge and $p$-charge is additive, is given in Appendix~\ref{app:computational}.  Concretely, $3|\delta_p|$ is the minimal energy required of repairing the residual $p$-charge using the two denominator directions $p^{\pm m}x^{3m}$, while the second term is the minimal energy required of repairing the two $A_2$ charge differences using the three numerator directions $y_1 x^2$, $y_2 x^2$, $(y_1 y_2)^{-1}x^2$.  These lower bounds add because the $p$-charge and $y$-charge repairs are carried by different letter types.

\paragraph{Tail survival under equal-charge projection.}
For $U(3)$, the BH existence boundary from~\cite{PZ} eq.~(4.7) is
\begin{equation}
  J_R^{\rm BH}(j)
  =\frac16\sqrt{j^2 + 2j\!\left(9-2\sqrt{2j+9}\right)
  +18\!\left(3-\sqrt{2j+9}\right)}.
  \label{eq:JRBH}
\end{equation}
The tail beyond this bound, first observed at equal fugacity in~\cite{PZ}, survives the equal-charge projection.  The main evidence is concentrated into two sets of data: the complete $j=66$ slice and selected high-spin examples.

\begin{figure}[htbp]
\centering
\includegraphics[width=0.85\textwidth]{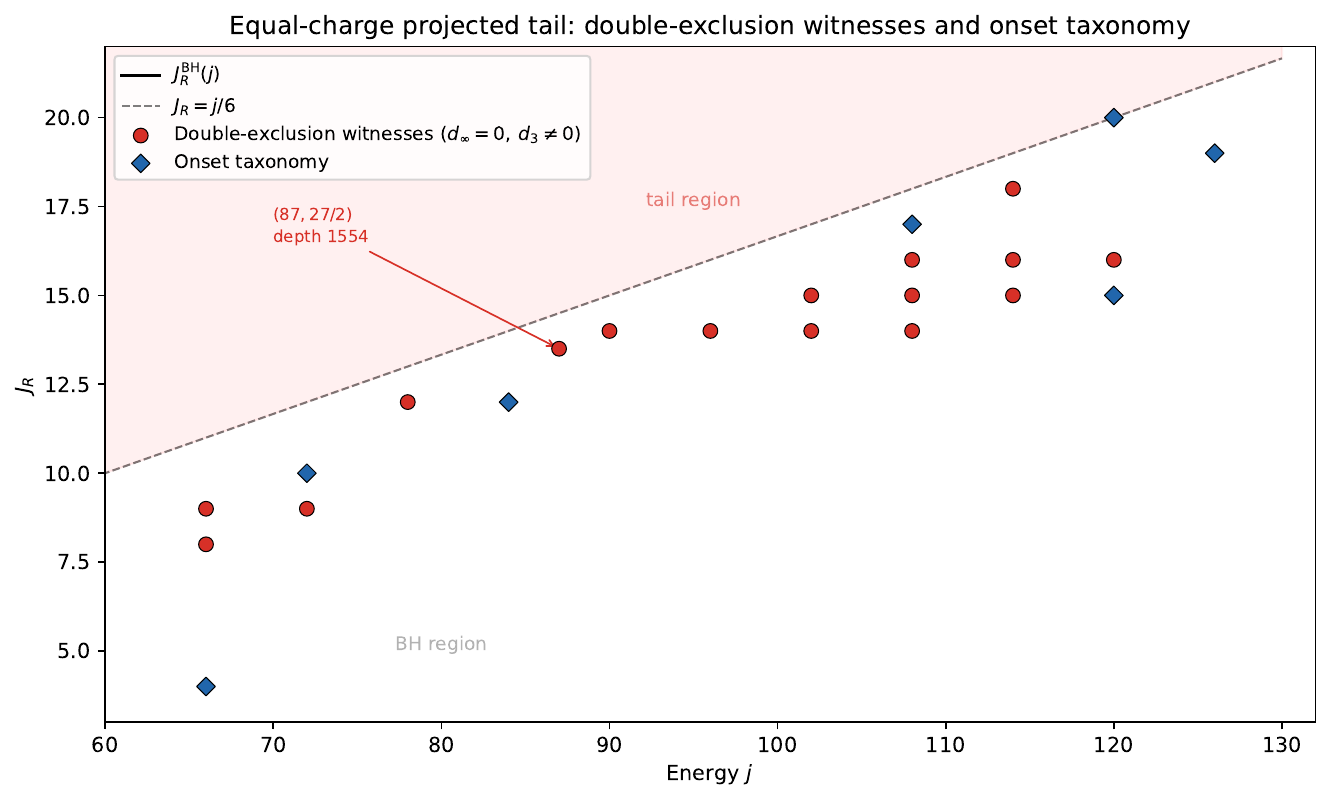}
\caption{Support map of the equal-charge projected tail in the selected high-spin range.  The black curve is the BH boundary $J_R^{\rm BH}(j)$; the dashed line is the kinematic ceiling $J_R=j/6$.  Red circles mark finite-rank coefficients beyond the BH bound with $d_\infty^{\rm eqQ}=0$ and $d_3^{\rm eqQ}\ne0$.  Blue diamonds mark the large-$N$ onset $j^*(J_R)$.}
\label{fig:supportmap}
\end{figure}

\subsection{\texorpdfstring{The complete $j=66$ slice}{The complete j=66 slice}}
\label{sec:j66}

The factorization theorem implies that $d_\infty^{\rm eqQ}(j,4)=0$ for all $j<66$, because the $[p^8]$ coefficient of the prefactor $\Pi(x,p)$ first becomes nonzero at $x^{66}$.  This requires a pair of generalized pentagonal numbers differing by~$8$; the first such pair is $\{7,15\}$, with $7=\omega(-2)$ and $15=\omega(-3)$, contributing at $x^{3(7+15)}=x^{66}$.  At $j=66$ itself, $d_\infty^{\rm eqQ}(66,4)=-1$.  For comparison, the equal-fugacity index at this charge is dominated by BH-scale coefficients~\cite{PZ}.

Using the pruning bound~\eqref{eq:pruning}, we compute the complete $j=66$ phase-space slice exactly:

\begin{table}[htbp]
\caption{Complete $j=66$ phase-space slice.  The table compares exact $U(3)$ projected coefficients with the multigraviton values from eq.~\eqref{eq:formula} and the correction $G_3^{\rm eqQ}:=d_3^{\rm eqQ}-d_\infty^{\rm eqQ}$.  The BH bound is $J_R^{\rm BH}(66)\simeq7.905$.}
\label{tab:j66full}
\begin{center}
\small
\begin{tabular}{r|r|r|r|l}
\toprule
$J_R$ & $d_3^{\rm eqQ}$ & $d_\infty^{\rm eqQ}$
& $G_3^{\rm eqQ}$ & comment \\
\midrule
$0$ & $270{,}582{,}770$ & $278{,}160$ & $270{,}304{,}610$ & BH \\
$1$ & $229{,}522{,}888$ & $-73{,}745$ & $229{,}596{,}633$ & BH \\
$2$ & $138{,}442{,}446$ & $-10{,}648$ & $138{,}453{,}094$ & BH \\
$3$ & $57{,}041{,}198$ & $-3{,}375$ & $57{,}044{,}573$ & BH \\
$4$ & $14{,}703{,}319$ & $-1$ & $14{,}703{,}320$ & BH; $d_\infty^{\rm eqQ}$ at onset \\
$5$ & $1{,}900{,}704$ & $124$ & $1{,}900{,}580$ & BH \\
$6$ & $22{,}938$ & $-343$ & $23{,}281$ & BH \\
$7$ & $-12{,}595$ & $27$ & $-12{,}622$ & near BH boundary \\
$8$ & $-9$ & $0$ & $-9$ & beyond BH; $j^*(8)=258$ \\
$9$ & $-26$ & $0$ & $-26$ & beyond BH; $j^*(9)=186$ \\
$10$ & $0$ & $0$ & $0$ & $j^*(10)=72>66$ \\
$11$ & $1$ & $1$ & $0$ & $d_\infty^{\rm eqQ}$ at onset \\
\bottomrule
\end{tabular}
\end{center}
\end{table}

Since $J_R^{\rm BH}(66)\simeq7.905$, the charge sectors $J_R=8,9,10,11$ lie beyond the BH bound.  The $J_R=8$ and $J_R=9$ rows of Table~\ref{tab:j66full} have all three properties: they lie beyond the BH bound, their multigraviton coefficients vanish (with $j^*(8)=258$ and $j^*(9)=186$), and their finite-$N$ coefficients are nonzero.

At $J_R=11$, by contrast, the coefficient $d_3^{\rm eqQ}(66,11)=1$ coincides with $d_\infty^{\rm eqQ}(66,11)=1$, so the aggregate giant-graviton correction vanishes: this beyond-BH coefficient is already accounted for by the large-$N$ multigraviton coefficient.

In the BH regime ($J_R\le7$), the finite-rank correction dominates $d_\infty^{\rm eqQ}$ in magnitude by orders of magnitude, ranging up to seven orders near the $J_R=4$ multigraviton onset ($d_\infty^{\rm eqQ}=-1$, $d_3^{\rm eqQ}=14{,}703{,}319$).

Sign inversion occurs at $J_R=7$: the finite-$N$ coefficient $d_3^{\rm eqQ}(66,7)=-12{,}595$ is negative while the multigraviton value $d_\infty^{\rm eqQ}(66,7)=+27$ is positive, consistent with partial cancellation between the multigraviton and finite-rank contributions near the BH boundary.

\subsection{\texorpdfstring{Selected high-spin extractions}{Selected high-spin extractions}}
\label{sec:targeted}

To test whether nonzero finite-rank coefficients persist beyond the BH bound outside the complete $j=66$ slice, we performed selected coefficient extractions in the high-spin region.  The integer-multiple-of-six slices $j=72,78,84,90,96,102,108$ were computed across the full beyond-BH band, while $j=87,114,120$ were selected to probe cases far below the large-$N$ onset.

\begin{table}[htbp]
\caption{Examples used in the main argument.  Each row has $J_R>J_R^{\rm BH}(j)$, $d_\infty^{\rm eqQ}=0$, and $d_3^{\rm eqQ}\ne0$.  The depth is $\Delta_\infty=j^*(J_R)-j$.  The rows with $j=6J_R+6$, namely
$(78,12)$, $(87,27/2)$, $(90,14)$, and $(114,18)$, are covered by Proposition~\ref{prop:strip}; the remaining rows probe farther inside the allowed region.}
\label{tab:main-examples}
\begin{center}
\footnotesize
\begin{tabular}{r|r|r|r|r}
\toprule
$j$ & $J_R$ & $d_3^{\rm eqQ}$ & $j^*(J_R)$ & $\Delta_\infty$ \\
\midrule
$66$ & $8$ & $-9$ & $258$ & $192$ \\
$66$ & $9$ & $-26$ & $186$ & $120$ \\
$72$ & $9$ & $-3$ & $186$ & $114$ \\
$78$ & $12$ & $-1$ & $84$ & $6$ \\
$87$ & $27/2$ & $1$ & $1641$ & $1554$ \\
$90$ & $14$ & $-1$ & $126$ & $36$ \\
$96$ & $14$ & $27$ & $126$ & $30$ \\
$102$ & $14$ & $93$ & $126$ & $24$ \\
$102$ & $15$ & $9$ & $120$ & $18$ \\
$108$ & $14$ & $-5{,}648$ & $126$ & $18$ \\
$108$ & $15$ & $-162$ & $120$ & $12$ \\
$108$ & $16$ & $-1$ & $1026$ & $918$ \\
$114$ & $15$ & $255$ & $120$ & $6$ \\
$114$ & $16$ & $218$ & $1026$ & $912$ \\
$114$ & $18$ & $-1$ & $198$ & $84$ \\
$120$ & $16$ & $317$ & $1026$ & $906$ \\
\bottomrule
\end{tabular}
\end{center}
\end{table}

The example set is therefore not confined to the first complete slice.  The row $(87,27/2)$ probes the largest value of $j^*(J_R)-j$ in Table~\ref{tab:main-examples}, while the $J_R=16$ examples give a second cluster with depths near $900$.

The kinematic-boundary coefficient is rank-independent: Proposition~\ref{prop:frontier} gives
\[
d_N^{\rm eqQ}(j,j/6)=E(j/3)
\]
for all $N$, and Proposition~\ref{prop:strip} determines the next line $j=6J_R+6$.  Outside these two lines there is no comparable general exact formula for the finite-$N$ coefficients at onset.  Three representative outcomes are
\[
\begin{array}{rcll}
  (72,10):  & d_\infty^{\rm eqQ}=-1, & d_3^{\rm eqQ}=0, & \text{exact cancellation},\\[2pt]
  (126,19): & d_\infty^{\rm eqQ}=1,  & d_3^{\rm eqQ}=-26, & \text{sign reversal},\\[2pt]
  (126,14): & d_\infty^{\rm eqQ}=-1, & d_3^{\rm eqQ}=523{,}465{,}399, & \text{large finite-rank coefficient}.
\end{array}
\]
Thus the formulas proved here control $J_R=j/6$ and $j=6J_R+6$, but they do not give a rank-independent formula for general interior coefficients.

Finite-rank dominance is not restricted to energies where the large-$N$ coefficient is zero.  For $J_R=2$, the second pentagonal pair first enters at $(j,J_R)=(144,2)$, where
\[
  d_\infty^{\rm eqQ}(144,2)=-496{,}793{,}087,
  \qquad
  d_3^{\rm eqQ}(144,2)=-116{,}665{,}457{,}868{,}988{,}568,
\]
so $|d_3/d_\infty|\simeq2.35\times10^8$ even where the large-$N$ multigraviton coefficient is already large.

The same effect appears after summing over all spin sectors.  Writing
\[
  D_N(j):=d_N^{\rm eqQ}(j,0)+2\sum_{J_R>0}d_N^{\rm eqQ}(j,J_R)
\]
for the equal-charge coefficient at $p=1$, the $p=1$ totals described in Appendix~\ref{app:computational} give $|D_3(j)/D_\infty(j)|\simeq480$ at $j=48$ and $\simeq6.1\times10^5$ at $j=90$.  Thus at rank $N=3$ the finite-rank corrections increasingly dominate the multigraviton sector as the energy grows in the computed range.

\subsection{\texorpdfstring{The $U(1)$ and $SU(3)$ questions}{The U(1) and SU(3) questions}}
\label{sec:U1-center}

All computations in this paper use the $U(N)$ index, retaining the free $U(1)$ center-of-mass mode.  At large $N$ this distinction decouples from the leading entropy~\cite{KMMR}, but at $N=3$ it can contribute at subleading order.  The $U(1)$ projected index is
\begin{equation}
  \mc I_1(x,p,y_1,y_2)
  =\exp\!\left[\sum_{m\ge1}
  \frac{f_m^{\rm ref}(x,p,y_1,y_2)}{m}\right],
  \label{eq:I1-PE}
\end{equation}
\noindent and it is not the same formal series as the multigraviton index $\mc I_\infty(x,p,y_1,y_2)$.  It has support in some of the same sectors as the finite-$N$ tail; for example
\[
  d_1^{\rm eqQ}(30,4)=-6,
  \qquad d_1^{\rm eqQ}(36,4)=-35,
  \qquad d_1^{\rm eqQ}(66,8)=193.
\]
Thus the free center cannot be ignored at $N=3$.

The free center can, however, be factored out exactly.  Because the $U(1)$
acts trivially on the adjoint, the refined index factorizes,
\begin{equation}
  \mc I_{U(3)}(x,p,y_1,y_2)
  \;=\;
  \mc I_{U(1)}(x,p,y_1,y_2)\,
  \mc I_{SU(3)}(x,p,y_1,y_2),
  \label{eq:U1SU3-factorization}
\end{equation}
as is manifest in~\eqref{eq:I3full}: $\mc I_3=K(1)\cdot\bigl[\tfrac16
K(1)^2\,\CT_{u,v}(\cdots)\bigr]$, and the bracket is the $SU(3)$ index, since
$\operatorname{adj}_{SU(3)}$ has zero-weight multiplicity two while the $A_2$
root weights, the Vandermonde factor $W(u,v)$, and the constant-term extraction are exactly
those of~\eqref{eq:I3full}.

The quotient by the common $\mathbb Z_3$ center
in $U(3)\simeq (U(1)\times SU(3))/\mathbb Z_3$ does not affect the index here because all
letters are in the adjoint and the center acts trivially.  Two consequences are
worth recording.  First,
the zero-weight replacement $K(1)^3\to K(1)^2$ used below computes the genuine
$SU(3)$ superconformal index, not merely an adjoint-character variant; the
qualification attached to Table~\ref{tab:su3-adjoint-check} concerns coverage
and the extent of the check, not the identity of the observable.

Equation~\eqref{eq:U1SU3-factorization} also supplies an end-to-end consistency check
that bypasses the direct $U(3)$ constant-term computation entirely: charge-resolved $SU(3)$ profiles
convolved with the explicitly known free series $K(1)$ must reproduce every
$U(3)$ value reconstructed from modular arithmetic.  The constant-term extraction does not distribute over the
product, so the equal-charge projections themselves do not factorize; the
convolution must be performed at the charge-resolved level before projecting.

We also checked selected coefficients in the interacting $SU(3)$ sector.  The adjoint-character identity $\operatorname{adj}_{U(3)}=\operatorname{adj}_{SU(3)}\oplus\mathbf 1$ changes the zero-weight multiplicity from three to two while leaving the six $A_2$ root weights, the Vandermonde factor $W(u,v)$, and the constant-term extraction unchanged.  In the notation of the computation this is the replacement $K(1)^3\to K(1)^2$.

An implementation using modular arithmetic on a CPU first reproduced the $U(3)$ checks
\begin{equation}
  c_3(30,0;0,0)=2628,\qquad
  c_3(30,0;1,-1)=c_3(30,0;2,1)=1970,
\end{equation}
and the examples with $d_\infty^{\rm eqQ}=0$, namely $d_3^{\rm eqQ}(66,8)=-9$, $d_3^{\rm eqQ}(87,27/2)=1$, and $d_3^{\rm eqQ}(108,16)=-1$.  Repeating selected extractions with the zero-weight multiplicity set to two gives the following coefficients.  By the factorization argument above these are genuine $SU(3)$ coefficients for the displayed targets, which we denote $d_{SU(3)}^{\rm eqQ}$; here ``selected'' means that only these entries were computed, not that a different observable was used.

These entries are reconstructed from two modular primes using the Chinese remainder theorem; the displayed integer is the representative with smallest absolute value modulo the product of the primes.

\begin{table}[htbp]
\caption{Genuine $SU(3)$ coefficients $d_{SU(3)}^{\rm eqQ}$ for selected examples with $d_\infty^{\rm eqQ}=0$.  Here $p_{\rm exp}=2J_R$, and ``charge perm.'' denotes the charge-permutation check $c_3(j,J_R;1,-1)=c_3(j,J_R;2,1)$.}
\label{tab:su3-adjoint-check}
\begin{center}
\footnotesize
\begin{tabular}{r|r|r|r|r|l}
\toprule
$j$ & $J_R$ & $p_{\rm exp}$ & $d_3^{\rm eqQ}$ & $d_{SU(3)}^{\rm eqQ}$ & check \\
\midrule
$30$  & $4$       & $8$  & $1$       & $-45$       & two primes \\
$66$  & $8$       & $16$ & $-9$      & $5854$      & two primes; charge perm. \\
$87$  & $27/2$    & $27$ & $1$       & $-406$      & two primes; charge perm. \\
$96$  & $14$      & $28$ & $27$      & $50{,}594$  & two primes \\
$108$ & $14$      & $28$ & $-5648$   & $-3{,}790{,}015$ & two primes \\
$108$ & $16$      & $32$ & $-1$      & $89{,}041$ & two primes; charge perm. \\
\bottomrule
\end{tabular}
\end{center}
\end{table}

The tested examples therefore do not disappear when the $U(1)$ zero-weight mode is removed.  In several cases the coefficient becomes much larger in magnitude: for example $d_3^{\rm eqQ}(108,16)=-1$ while $d_{SU(3)}^{\rm eqQ}(108,16)=89{,}041$, and $d_3^{\rm eqQ}(87,27/2)=1$ while $d_{SU(3)}^{\rm eqQ}(87,27/2)=-406$.  This disfavors the explanation that these coefficients are produced solely by the decoupled center.

In the light of eq.~\eqref{eq:U1SU3-factorization} this is the expected generic situation: a small $U(3)$ coefficient such as $d_3^{\rm eqQ}(87,27/2)=1$ results from cancellations in the convolution of the free-center series with the larger interacting-sector coefficient $d_{SU(3)}^{\rm eqQ}(87,27/2)=-406$, dressed by the free-center tower.  Thus the unit value understates the interacting-sector support.

We keep the main statement of the paper in the $U(3)$ language because the main scan is a $U(3)$ scan; the $SU(3)$ entries above are genuine coefficients for their displayed targets, not a full independent $SU(3)$ scan.

\section{Black-hole entropy and giant-graviton structure}
\label{sec:BH-physics}

The results of the preceding sections establish the arithmetic of the equal-charge projected index: which charge sectors have nonzero coefficients, for which sectors the large-$N$ coefficient is zero, and where finite-$N$ corrections dominate the multigraviton index.  We now turn to the physics: how the projected index compares with the black-hole entropy, and what the finite-$N$ corrections reveal about the underlying brane structure.

\subsection{\texorpdfstring{Entropy comparison at $j=66$}{Entropy comparison at j=66}}

At $j=66$, the classical AdS$_5$ BH entropy evaluated with the $N=3$ normalization~\cite{PZ,ACKK} gives
\begin{equation}
  S_{\rm BH}^{N=3}(66,J_R)
  \in\{24.28,\,24.10,\,23.53,\,22.55,\,21.08,\,19.00,\,16.05,\,11.52\},
  \label{eq:SBH66}
\end{equation}
for $J_R=0,1,\ldots,7$, with $S_{\rm BH}=0$ for $J_R\ge J_R^{\rm BH}(66)\simeq7.905$.

Figure~\ref{fig:entropy} compares $\log|d_3^{\rm eqQ}(66,J_R)|$ with $S_{\rm BH}$ across the full $J_R$ spectrum.  In the BH bulk ($J_R=0,\ldots,5$), the deficit
\begin{equation}
  \log|d_3^{\rm eqQ}(66,J_R)|-S_{\rm BH}^{N=3}(66,J_R)
  \approx -4.72
  \label{eq:deficit}
\end{equation}
is approximately constant, with maximum deviation $0.18$.

At this single slice, the projected finite-$N$ index has a bulk $J_R$-dependence roughly parallel to the classical entropy curve, with a fixed offset attributable to the restriction $Q_1=Q_2=Q_3$, cancellations caused by the $(-1)^F$ sign in the index, and small-$N$ effects.  At $J_R=6,7$ the tracking breaks down---the index drops faster than the entropy---and beyond $J_R^{\rm BH}$ the entropy vanishes while the index tail survives.  The multigraviton index $\log|d_\infty^{\rm eqQ}|$ lies far below both curves throughout the region inside the classical BH bound.  Thus at rank $N=3$ the finite-rank correction, not the multigraviton index, controls the magnitude $|d_3^{\rm eqQ}|$.

All coefficients studied here are $(-1)^F$-graded signed coefficients of the superconformal index.  We therefore use $|d|$ only to track exponential growth, not as a literal positive degeneracy.  This distinction is important at finite~$N$, where the positive BPS partition function~\cite{CaboBizetZBPS} rather than a single coefficient of the $(-1)^F$-graded index is the natural total-state-counting object.

\begin{figure}[t]
\centering
\includegraphics[width=0.85\textwidth]{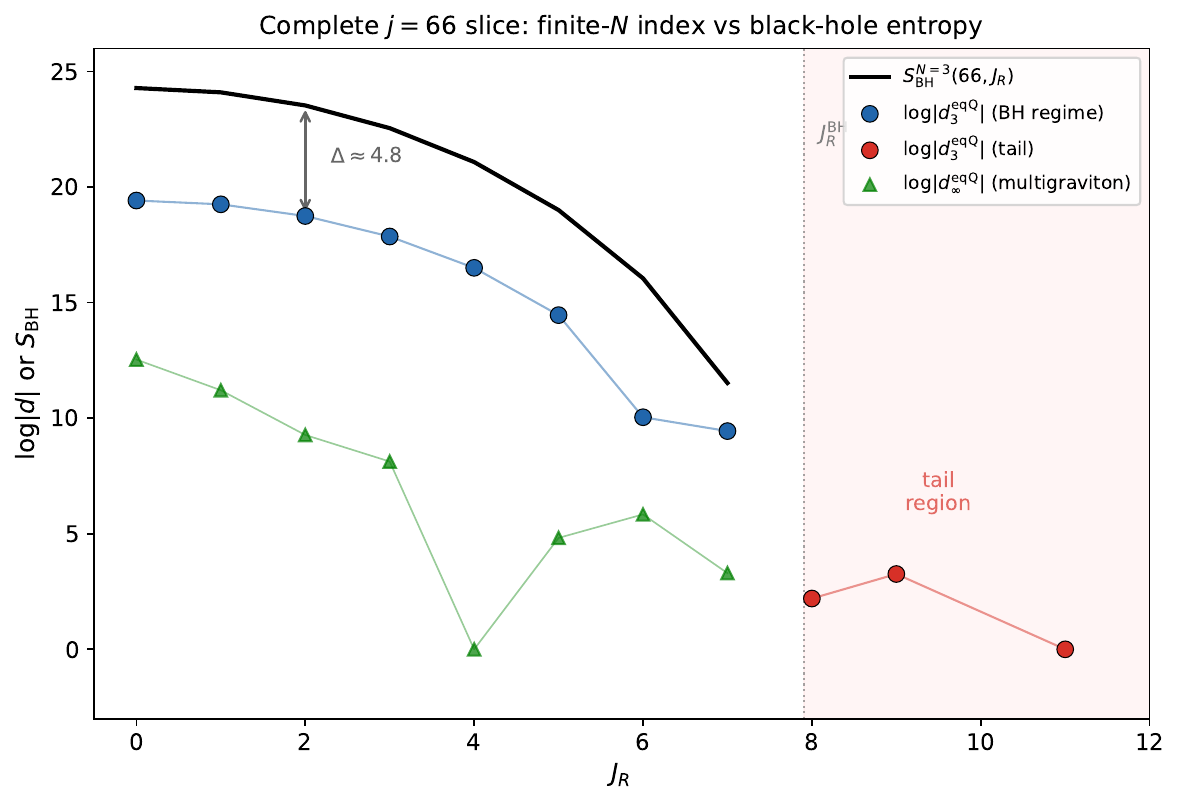}
\caption{Complete $j=66$ equal-charge slice.  The plot compares the classical BH entropy $S_{\rm BH}^{N=3}(66,J_R)$ with $\log|d_3^{\rm eqQ}|$ and $\log|d_\infty^{\rm eqQ}|$.  The finite-rank index roughly tracks the BH entropy in the bulk with an approximately constant deficit, while the tail survives beyond $J_R^{\rm BH}$.}
\label{fig:entropy}
\end{figure}

This comparison is more restrictive than the equal-fugacity analysis of~\cite{PZ,ACKK}, which tunes $y_1=y_2=1$ but retains all charge sectors.  The equal-charge projection isolates $Q_1=Q_2=Q_3$ microcanonically, matching the charge assignment of the supersymmetric rotating black hole~\cite{GR,CCLP}.  The near-constant deficit~\eqref{eq:deficit} is consistent with a local Gaussian profile, i.e.\ with restricting an approximately Gaussian distribution of charge differences to the equal-charge slice.  The charge-resolved large-$N$ formula makes this interpretation quantitative and gives the benchmark used below.

Physically, a Gaussian profile is the expected statistical mechanics of the
free multigraviton gas: the charge differences $(\Delta_1,\Delta_2)$ of a
multi-letter state are sums of many independent letter charges, so a
central-limit profile with computable width is natural.
Proposition~\ref{prop:charge-resolved} makes this expectation quantitative at
large $N$: expanding each active pentagonal-pair contribution gives the
Gaussian width below, while grades with strong signed cancellations can show
additional deviations.

Writing $L(c)=\log\pfn(c)$ and expanding the
profile~\eqref{eq:formula-resolved} at fixed $(j,J_R)$ to second order in
$(\Delta_1,\Delta_2)$ gives, on its support
sublattice,\footnote{With $\delta=(\Delta_1+\Delta_2)/3$, the three partition
arguments in~\eqref{eq:formula-resolved} are shifted by
$(-\delta,\,\Delta_1-\delta,\,\Delta_2-\delta)$ relative to the $\Delta=0$
profile at the same $(j,J_R)$; the shifts sum to zero and their squares sum to
$\tfrac23(\Delta_1^2+\Delta_2^2-\Delta_1\Delta_2)$,
giving~\eqref{eq:gaussian-profile} at second order.}
\begin{equation}
  \frac{d_\infty^{(\Delta_1,\Delta_2)}(j,J_R)}
       {d_\infty^{(0,0)}(j,J_R)}
  \;\simeq\;
  \exp\!\Bigl[\tfrac13\,L''(c)\,
  \bigl(\Delta_1^2+\Delta_2^2-\Delta_1\Delta_2\bigr)\Bigr],
  \qquad c\simeq j/6,
  \label{eq:gaussian-profile}
\end{equation}
a Gaussian whose variance along the $(\Delta,0)$ slice used below is
\begin{equation}
  \sigma^2=-\frac{3}{2L''(c)}=\frac{j^{3/2}}{2\pi}\bigl(1+O(j^{-1/2})\bigr)
  \label{eq:sigma}
\end{equation}
by the Hardy--Ramanujan asymptotics
$\log\pfn(c)=\pi\sqrt{2c/3}-\log(4\sqrt3\,c)+O(c^{-1/2})$~\cite{HR}; the
logarithmic term is the prefactor correction to $L''(c)$ used in the numerical
check below.  By the equal-charge suppression factor we mean
the coefficient-level ratio between the charge-summed large-$N$ coefficient and
its equal-charge constant term at fixed $(j,J_R)$; the
profile~\eqref{eq:gaussian-profile} with width~\eqref{eq:sigma} implies a power
law $\propto j^{3/2}$, not an exponential suppression.

The exact coefficients follow this prediction in the sampled regimes where the
signed pentagonal-pair cancellations are not anomalously large.  One-dimensional profile ratios are
useful checks, but at small $j$ their numerical values depend on how one
samples the exact $A_2$ support sublattice, and the quadratic form
in~\eqref{eq:gaussian-profile} makes one-dimensional variances
direction-dependent as well.

For the numerical checks we use a concrete one-dimensional slice of the
support sublattice, namely $(\Delta_1,\Delta_2)=(\Delta,0)$.  The slices
obtained from this one by the multiset symmetry of
eq.~\eqref{eq:charge-resolved} are equivalent.  Along this slice, the measured
variance agrees with~\eqref{eq:sigma} after including the Hardy--Ramanujan
prefactor correction in $L''(c)$; the
agreement is at the few-percent level, and is within $2\%$ at $j=600$.  The
integrated equal-charge suppression factor follows the same predicted power
law.

Since the numerator and denominator are coefficients of a $(-1)^F$-graded
index, the suppression factor measures suppression in the index, not a ratio of
positive degeneracies.  The exact value at $j=66$ (Table~\ref{tab:j66full}) is
\[
  \frac{[x^{66}p^0]\,\mc I_\infty(x,p)}{d_\infty^{\rm eqQ}(66,0)}
  =\frac{70{,}634{,}952}{278{,}160}\simeq254,
  \qquad
  \log254\simeq5.5,
\]
of the same size as the observed finite-$N$ deficit~\eqref{eq:deficit}.  The
spin-summed totals show the same scaling: with $\mc I_\infty(x,p)$ as
in~\eqref{eq:Iinf_eqfug},
a log--log least-squares fit of $\log\bigl([x^j]\,\mc I_\infty(x,1)/D_\infty(j)\bigr)$ against $\log j$ over the nonzero grades $600\le j\le3000$ gives slope $1.483$ with $R^2=0.999998$.

\FloatBarrier

\subsection{\texorpdfstring{Sector check at \((87,27/2)\)}{Sector check at (87,27/2)}}
\label{sec:GG-decomposition}

The giant-graviton expansion~\cite{GE,Imamura,LR,Murthy-GG,ABMM,CKLL-GG2BH,Beccaria-CB} organizes the finite-$N$ correction as
\begin{equation}
  \frac{\mc I_N(x,p,y_1,y_2)}{\mc I_\infty(x,p,y_1,y_2)}=1+\sum_{M\ge1}\mc G_N^{(M)},
  \label{eq:GG-expansion}
\end{equation}
\noindent with projected coefficient-level contributions
\begin{equation}
  h_{N,M}^{\rm eqQ}(j,J_R)
  =[x^jp^{2J_R}y_1^0y_2^0]~\mc I_\infty(x,p,y_1,y_2)\,\mc G_N^{(M)}.
  \label{eq:h-def}
\end{equation}
Thus
\begin{equation}
  d_N^{\rm eqQ}(j,J_R)-d_\infty^{\rm eqQ}(j,J_R)
  =\sum_{M\ge1}h_{N,M}^{\rm eqQ}(j,J_R).
  \label{eq:coef-GG-expansion}
\end{equation}
The present paper computes the aggregate correction $G_3^{\rm eqQ}$ exactly.  For most coefficients it does not separate $G_3^{\rm eqQ}$ into the individual $h_{3,M}^{\rm eqQ}$, and this distinction is important: small aggregate numbers can result from cancellations among several sectors rather than from a single wrapped-brane sector.

The following rank differences illustrate this point.  Writing the two rank steps as $(d_2^{\rm eqQ}-d_\infty^{\rm eqQ})$ and $(d_3^{\rm eqQ}-d_2^{\rm eqQ})$, for example,
\begin{align*}
  G_3^{\rm eqQ}(126,19)&=-27=58-85,\\
  G_3^{\rm eqQ}(84,12)&=28=-67+95,\\
  d_3^{\rm eqQ}(96,14)&=G_3^{\rm eqQ}(96,14)=27=-29+56.
\end{align*}
The third line is written as a $d_3^{\rm eqQ}$ value because $d_\infty^{\rm eqQ}(96,14)=0$ there.  These decompositions use $d_2^{\rm eqQ}-d_\infty^{\rm eqQ}$ and $d_3^{\rm eqQ}-d_2^{\rm eqQ}$; they are useful checks, but they are not the genuine $M$-giant sectors~\cite{Murthy-GG,LR}.

The genuine wrapped-brane test is to apply the sector formulae~\cite{LR} sector by sector with the equal-charge projection.

For the coefficient at $(87,27/2)$ this can in fact be checked for the first three sectors: the $M=1$ value by a direct exact computation over integers, and the $M=2,3$ values modulo the two large primes $p_1=10^9+7$ and $p_2=998244353$.  Applying the $M$-giant sector formulae~\cite{LR,Murthy-GG} after the equal-charge projection gives
\begin{equation}
  h_{3,1}^{\rm eqQ}(87,\tfrac{27}{2})=1\quad\text{exactly},
  \qquad
  h_{3,2}^{\rm eqQ}(87,\tfrac{27}{2})
  \equiv
  h_{3,3}^{\rm eqQ}(87,\tfrac{27}{2})
  \equiv0\pmod{p_1p_2}.
  \label{eq:h31-example}
\end{equation}
The one-giant sector is built from the determinant prefactor and the projected single-particle letters
\begin{equation}
  \mc G_3^{(1)}
  =\Big[-\frac{\zeta}{(1-\zeta)^2}\;
  {\rm PE}\!\left[(\zeta+\zeta^{-1}-2)\,\widehat\imath\,\right]\Big]_{\zeta^{-3}},
  \qquad
  \widehat\imath
  =\frac{(1-px^3)(1-p^{-1}x^3)}
  {(1-y_1x^2)(1-y_2x^2)(1-x^2/(y_1y_2))}-1,
  \label{eq:G31}
\end{equation}
where $[\,\cdot\,]_{\zeta^{-3}}$ extracts the wrapping-number-one contribution and $-\zeta/(1-\zeta)^2$ is the single-giant determinant factor of~\cite{LR}.  The first terms of $\widehat\imath$ are
\[
  \widehat\imath
  =x^2\!\left(y_1+y_2+\frac{1}{y_1 y_2}\right)
  -x^3(p+p^{-1})+O(x^4).
\]
To reach the target $x^9 p\,y_1^0 y_2^0\,\zeta^{-4}$ (the determinant factor supplies the remaining $\zeta^{+1}$) at minimal energy, one must choose the $\zeta^{-1}\widehat\imath$ part four times.  The $p$-charge requires the letter $-\zeta^{-1}x^3 p$; the remaining energy is $6$, and the only minimal equal-charge combination is $\zeta^{-1}x^2 y_1$, $\zeta^{-1}x^2 y_2$, $\zeta^{-1}x^2(y_1 y_2)^{-1}$.

Evaluating the PE coefficient gives
\begin{equation}
  [x^9\,p^1\,y_1^0\,y_2^0\,\zeta^{-4}]\,
  {\rm PE}\!\left[(\zeta+\zeta^{-1}-2)\,\widehat\imath\,\right]=-1.
  \label{eq:giant-coeff}
\end{equation}
The remaining two factors are
\[
  [\zeta^1]\!\left(-\frac{\zeta}{(1-\zeta)^2}\right)=-1,
  \qquad
  [x^{78}p^{26}]\,\mc I_\infty^{\rm eqQ}(x,p)=+1.
\]
Hence $h_{3,1}^{\rm eqQ}=(-1)(-1)(+1)=1$.  Concretely, the coefficient~\eqref{eq:giant-coeff} is carried by the four lowest letters,
\begin{equation}
  (\zeta^{-1}x^2y_1)(\zeta^{-1}x^2y_2)(\zeta^{-1}x^2y_1^{-1}y_2^{-1})(-\zeta^{-1}x^3p)
  =-\zeta^{-4}x^9p.
  \label{eq:four-letters}
\end{equation}
The three $x^2$ letters carry charges $y_1$, $y_2$, and $(y_1y_2)^{-1}$, so they form the equal-charge monomial $y_1^0y_2^0$ at energy $2N=6$.  The fourth letter supplies the single unit of $p$.  Together with the large-$N$ factor $[x^{78}p^{26}]\mc I_\infty^{\rm eqQ}(x,p)=1$, this is the minimal one-giant mechanism behind the coefficient at $(87,27/2)$.  No other term in the plethystic expansion can reach the target $x^9p\,y_1^0y_2^0\,\zeta^{-4}$: the energy budget $9=3\times2+3$ and the charge constraints uniquely fix the four letters above.

The exact equality $h_{3,1}^{\rm eqQ}=1$ is the explicit one-giant computation displayed above.  The congruences for $M=2,3$ in~\eqref{eq:h31-example} were obtained by direct implementations that store only nonzero monomials of the genuine $M$-giant $U(M)$ residue formula~\cite{LR}; they are not higher-$\zeta$ coefficients of the one-giant expression~\eqref{eq:G31}.

As further low-energy checks, the same sector implementation reproduces aggregate corrections with nontrivial two-giant contributions, for example
\begin{equation}
  G_3^{\rm eqQ}(24,1)=-426+180=-246,
  \qquad
  G_3^{\rm eqQ}(24,2)=-50+5=-45.
  \label{eq:sector-regression-checks}
\end{equation}

The modular congruences provide strong evidence for the vanishing of the first two higher sectors.  Independently, the exact aggregate relation gives
\begin{equation}
  \sum_{M\ge2}h_{3,M}^{\rm eqQ}(87,\tfrac{27}{2})=0.
  \label{eq:higher-M-remainder}
\end{equation}
Modulo $p_1 p_2$, the sector checks further imply $\sum_{M\ge4}h_{3,M}^{\rm eqQ}\equiv0\pmod{p_1 p_2}$.

The coefficient at $(87,\tfrac{27}{2})$ is therefore not only the result of taking a rank difference: in the $M$-giant sector expansion~\cite{LR}, it is already present in the $M=1$ sector, and the first two higher sectors are zero modulo $p_1p_2$.

Two caveats are kept deliberately.  First, equations~\eqref{eq:h31-example} and~\eqref{eq:higher-M-remainder} do not establish that $h_{3,M}^{\rm eqQ}(87,\tfrac{27}{2})=0$ for each $M\ge2$ individually.  They establish the exact value $h_{3,1}=1$ and the exact aggregate identity that the remaining higher sectors sum to zero.

The modular vanishing of $h_{3,2}$ and $h_{3,3}$ is strong evidence for those two terms, but an independent coefficient bound would be needed to turn the modular congruences into exact integer vanishings.  We do not have a kinematic argument bounding the wrapping number: the only energy constraint we can prove at finite $N$ is the rank-independent kinematic bound $j\ge6|J_R|$ of~\eqref{eq:kinematic}, which does not separate the wrapping sectors.  The support congruence $j-6J_R\in6\mathbb Z_{\ge0}$ already makes the first available energy step above the kinematic floor equal to $6$, which likewise does not separate the wrapping sectors.

Second, as emphasized in~\cite{LR}, the matrix-model sectors are not unique as a physical wrapped-D3 decomposition once $M\ge2$; the statement here is at the level of those matrix-model sectors.  For larger coefficients and near-boundary examples the warning about cancellations above therefore remains in force.  The second target,
\[
  523{,}465{,}400=\sum_{M\ge1}h_{3,M}^{\rm eqQ}(126,14),
\]
has not yet been decomposed into its individual $M$-giant contributions and remains a natural next step.

\section{Discussion and conclusions}
\label{sec:discussion}

The equal-charge projection shows that the high-spin tail is not an artifact of summing over unequal $R$-charge assignments.  The large-$N$ side is completely controlled by the factorization theorem.  By Theorem~\ref{thm:support}, for fixed $J_R$ the multigraviton coefficient obeys an arithmetic support criterion and vanishes throughout
\[
  6J_R\le j<j^*(J_R).
\]
Therefore a nonzero finite-rank coefficient in this interval is not supplied by the large-$N$ multigraviton index.  It is a finite-rank contribution to the equal-charge index, which is $(-1)^F$-graded.

The exact mechanisms are different.  After the plethystic exponential has formed the multigraviton product, the intervals with zero large-$N$ coefficient come from diagonal projection followed by the fact that the large-$N$ prefactor is supported only at pentagonal exponents (Theorem~\ref{thm:main} and eq.~\eqref{eq:formula}).  The rank-independent kinematic boundary comes instead from saturation of a nonnegative single-letter defect: at $j=6J_R$ every letter must have zero defect, reducing the finite-$N$ matrix integral to the standard constant-term identity used in Proposition~\ref{prop:frontier}.  The same defect expansion also gives the line $j=6J_R+6$, through Proposition~\ref{prop:strip}.  Thus the lines $j=6J_R$ and $j=6J_R+6$ are controlled analytically here, while finite-rank corrections can be nonzero in the remaining interior intervals (Table~\ref{tab:main-examples}).

The $U(3)$ computation supplies such coefficients beyond the black-hole bound.  The examples in Table~\ref{tab:main-examples} have
\[
  J_R>J_R^{\rm BH}(j),\qquad
  d_\infty^{\rm eqQ}(j,J_R)=0,
  \qquad
  d_3^{\rm eqQ}(j,J_R)\ne0.
\]

The entry $(j,J_R)=(87,27/2)$ lies far below the multigraviton onset determined by the pentagonal-number condition. It also lies on
$j=6J_R+6$, so its value is the $P=27$ case of Proposition~\ref{prop:strip}. This large depth is arithmetic: it reflects the delayed multigraviton onset rather than proximity to a large-$N$ threshold.  The nearby $j=108$ band gives a complementary pattern: three consecutive beyond-boundary spin sectors have nonzero finite-rank coefficients while the large-$N$ coefficient is zero, followed by a matched onset and the rank-independent kinematic boundary.  Thus the evidence is not only a single extreme coefficient, but also a pattern of nonzero coefficients in neighboring sectors.

The sector check gives the first wrapped-brane refinement of this support statement.  For the coefficient at $(87,27/2)$,
\[
  h_{3,1}^{\rm eqQ}(87,\tfrac{27}{2})=1,\qquad
  h_{3,2}^{\rm eqQ}(87,\tfrac{27}{2})
  \equiv
  h_{3,3}^{\rm eqQ}(87,\tfrac{27}{2})
  \equiv0\pmod{p_1p_2},
\]
with the $M=2,3$ values obtained as two-prime modular $M$-giant sector checks~\cite{LR,Murthy-GG}.  Thus the unit aggregate coefficient is visible in the one-giant sector; it is not only the result of taking a rank difference.  The exact aggregate identity gives
\[
  \sum_{M\ge2} h_{3,M}^{\rm eqQ}(87,\tfrac{27}{2})=0.
\]
Thus the total higher-sector contribution vanishes, although the individual
higher-$M$ terms are not all determined separately by the present computation.
For larger coefficients, such as $G_3^{\rm eqQ}(126,14)=523{,}465{,}400$, the individual $M$-giant contributions~\cite{LR} have not yet been computed.

These conclusions should not be read as formulae for positive degeneracies or entropy.  All coefficients in this paper are coefficients of a $(-1)^F$-graded index~\cite{KMMR}.  They are protected signed quantities, not literal counts of BPS states.  The entropy comparison at $j=66$ therefore compares the exponential scale and spin dependence, not a microscopic state count.  The point established here is narrower: the equal-charge index has finite-rank support in sectors where the multigraviton index vanishes by the pentagonal onset mechanism of Theorem~\ref{thm:support}.

There is one group-theoretic point to keep separate from the main claim.  The main scan is a $U(3)$ scan.  At the same time, because the common center acts trivially on adjoint letters, the refined index factorizes as $\mc I_{U(3)}=\mc I_{U(1)}\mc I_{SU(3)}$.  Therefore the entries in Table~\ref{tab:su3-adjoint-check} are genuine $SU(3)$ coefficients for the displayed targets.  What is missing is not correctness but coverage: we do not perform a full independent $SU(3)$ scan.

The paper contains analytic results, exact computations, and partial sector checks.  Analytically, we prove the equal-charge large-$N$ factorization and its charge-resolved refinement, the onset criterion and the intervals where the large-$N$ coefficient vanishes, the rank-independence of the kinematic boundary, and the formula on the first adjacent line $j=6J_R+6$.  The two-sided support characterization (Theorem~\ref{thm:support}) is proved whenever at most three pentagonal pairs contribute, and is verified computationally through $j=9000$ in the remaining cases.  Computationally, we determine the displayed $U(3)$ coefficients exactly, and Table~\ref{tab:su3-adjoint-check} gives genuine $SU(3)$ coefficients for selected targets.  At the sector level, the one-giant value $h_{3,1}^{\rm eqQ}(87,\tfrac{27}{2})=1$ is explicit, while the $M=2,3$ vanishings are modular checks.  Broader wrapped-brane interpretations, uncomputed higher-$M$ decompositions, and near-boundary support patterns should therefore be read as observations suggested by the data rather than proved statements.

In summary, the high-spin tail survives the exact equal-charge projection, and some of its finite-rank support lies in sectors where the large-$N$ multigraviton coefficient is exactly zero.  Several concrete tests would strengthen these results.  The congruences $h_{3,2}^{\rm eqQ}(87,\tfrac{27}{2})\equiv0$ and $h_{3,3}^{\rm eqQ}(87,\tfrac{27}{2})\equiv0$ at $(87,\tfrac{27}{2})$ could be turned into exact integer statements by proving an upper bound on their absolute values, by carrying out an exact integer computation, or by combining enough modular primes with such a bound.  More individual higher-$M$ equal-charge sectors could be computed~\cite{LR}, and the same test repeated at $U(4)$.  A quantitative comparison with the grey-galaxy phase diagram of~\cite{CJKKLMP,Bajaj-grey} is not immediate, because that analysis fixes the average charge and sums over charge differences, whereas the present projection fixes the charge differences and sums over the common charge; mapping the exact equal-charge tail onto their microcanonical phase boundaries is left for future work.  A full independent $SU(3)$ scan, and a comparison of the same sectors with a positive BPS counting observable~\cite{CaboBizetZBPS,CL-words,CCKLL}, would substantially sharpen the picture.  The analytic framework developed here also suggests systematic finite-rank and sector-by-sector extensions.

\bigskip
\noindent\textbf{Acknowledgements.}
The author is grateful to Leopoldo Pando Zayas for explaining aspects of his results and to Leonardo Santilli for discussions on these topics and valuable comments on the draft.

\appendix
\section{Computational details}
\label{app:computational}

All finite-$N$ index coefficients reported here are integers, computed by multiplying Laurent polynomials stored by their nonzero monomials.  The polynomial multiplications are performed in C, with a Python script used only to launch the runs and collect outputs; the C routine uses checked 64-bit arithmetic over integers.  The direct run over integers for $(j,J_R)=(126,14)$ overflowed during the $P_{\rm TERMS}$ convolution, so this coefficient was recomputed by modular arithmetic.  In the supplementary GPU files, \texttt{p\_exp} denotes the exponent of the angular fugacity~$p$, so \texttt{j=126,p\_exp=28} corresponds to $(j,J_R)=(126,14)$.  We used five independent 30-bit primes
\begin{equation}
  998244353,\;\; 999999937,\;\; 1000000007,\;\; 1000000009,\;\; 1004535809.
\end{equation}
Chinese-remainder reconstruction from four primes (product $\approx1.0\times10^{36}$) gives
\begin{equation}
  d_3^{\rm eqQ}(126,14)=523{,}465{,}399,
\end{equation}
the representative of smallest absolute value modulo the four-prime product; the fifth prime independently confirms this value, extending the combined modulus to approximately $10^{45}$.  The associated multigraviton coefficient is $d_\infty^{\rm eqQ}(126,14)=-1$, hence $G_3^{\rm eqQ}(126,14)=523{,}465{,}400$.  The same five-prime run also reconstructs the sampled charge-difference profile around this point and passes the charge-permutation check $c_3(126,14;1,-1)=c_3(126,14;2,1)$ prime by prime.

The factorization theorem has been independently verified through $j=120$ by a direct constant-term extraction on the refined denominator product, bypassing the $U(3)$ projector entirely.  Independent consistency checks include the $p=1$ reconstruction
\begin{equation}
  d_3^{\rm eqQ}(66,0)
  =D_3(66)-2\sum_{J_R=1}^{11}d_3^{\rm eqQ}(66,J_R)
  =270{,}582{,}770,
\end{equation}
the matching of all equal-fugacity coefficients (where $y_1=y_2=1$ is an independent computation), the onset match $d_3^{\rm eqQ}(108,17)=d_\infty^{\rm eqQ}(108,17)=1$ (the $J_R=17$ onset of Table~\ref{tab:onset}, $j^*(17)=108$; the kinematic boundary at $j=108$ is $J_R=18$), and the kinematic-boundary zero $d_3^{\rm eqQ}(108,18)=0=E(36)$ (Proposition~\ref{prop:frontier}, using~\cite{Macdonald,ZB}).

The high-$j$ charge-resolved extractions use a GPU implementation using modular arithmetic over the same five primes.  Four primes are used for Chinese-remainder reconstruction and the fifth for independent verification.  For every run containing both $\Delta=(1,-1)$ and $\Delta=(2,1)$, the charge-permutation equality $c_3(j,J_R;1,-1)=c_3(j,J_R;2,1)$ is checked prime by prime before including the coefficient.  No charge-profile coefficient failing either the charge-permutation check or the modular verification is included in the reported data set.

The $p=1$ projected total $D_3(j)$ has been computed through $j=90$; the multigraviton index $D_\infty(j)$ has been independently computed from the closed-form product~\eqref{eq:Dinf} and verified against summation over the fixed-$J_R$ data at $j=66$.  Reproducible code and selected exact integer data for the displayed tables, together with the grade list for the $j\le9000$ scan of Theorem~\ref{thm:support}, are included as ancillary files with the arXiv submission.

The genuine $SU(3)$ coefficients of Section~\ref{sec:U1-center} were obtained with a small CPU implementation using modular arithmetic and the same recurrence and pruning logic.  In that port the parameter \texttt{diag} is the zero-weight multiplicity in the adjoint character: \texttt{diag=3} reproduces the $U(3)$ calculation, while \texttt{diag=2} computes the genuine $SU(3)$ coefficients reported in Table~\ref{tab:su3-adjoint-check}.  The choice \texttt{diag=1} should not be identified with the standalone $U(1)$ index or with the multigraviton index, because the six $A_2$ root weights, the Vandermonde factor $W(u,v)$, and the constant-term extraction are still present.

\begin{proof}[Proof of Lemma~\ref{lem:pruning}]
For the $y$-charges, introduce nonnegative integers $n_1,n_2,n_3$ counting minimal letters in the three numerator directions $(y_1,y_2,(y_1y_2)^{-1})$, each with energy requirement~$2$.  The residual $y$-charge equations $n_1-n_3=\delta_1$, $n_2-n_3=\delta_2$ are solved by $n_1=\delta_1+c$, $n_2=\delta_2+c$, $n_3=c$ with minimal $c=\max(-\delta_1,-\delta_2,0)$, giving the $y$-charge energy requirement $2(\delta_1+\delta_2+3\max(-\delta_1,-\delta_2,0))$.
For the $p$-charge, write $r-s=\delta_p$ with $r,s\ge0$; the energy requirement is $3(r+s)\ge3|\delta_p|$.  Since the $p$-charge comes from the denominator directions and the $y$-charge from the numerator directions, the lower bounds are additive.  Any completion exceeding the remaining energy budget $j-e$ is therefore impossible.
\end{proof}

\medskip\noindent\textit{Support congruence at finite $N$.}
A monomial in the $m$-th single-letter contribution has the form
\[
  x^{2m(n_1+n_2+n_3)+3m(r+s)}\,p^{m(r-s)}\,y_1^{m(n_1-n_3)}\,y_2^{m(n_2-n_3)}
\]
with $n_i\in\{0,1\}$ and $r,s\ge0$.  Writing $\Delta_i=Q_i-Q_3$ and $P=2J_R$, every such monomial satisfies $j-3P-2(\Delta_1+\Delta_2)\in 6\mathbb Z$.  The congruence is additive under multiplication and therefore holds for every monomial in the finite-$N$ plethystic expansion.  In the equal-charge sector $\Delta_1=\Delta_2=0$, so $j\equiv 6J_R\pmod{6}$, confirming that $d_N^{\rm eqQ}(j,J_R)$ can be nonzero only when $j-6J_R\in6\mathbb Z$.

\section{Low-degree constant terms with insertions}
\label{app:inserted-constant-terms}

This appendix evaluates the normalized constant terms with insertions
$S_1$, $S_1^2$, and $S_1^3$ used in the proof of
Proposition~\ref{prop:strip}.

\begin{lemma}[Low-degree inserted constant terms]
\label{lem:inserted-constant-terms}
Let
\[
  Z_N(q)=\frac{1}{N!}\CT_{\mathbf z}
  \prod_{i\ne j}(z_i/z_j;q)_\infty
  =(q;q)_\infty^{1-N},
\]
and define the normalized expectation $\langle\cdot\rangle_N$ with
respect to this measure.  With
$S_1=\sum_{a,b=1}^N z_a/z_b$,
one has
\[
  \langle S_1\rangle_N=1-q\qquad(N\ge2),
\]
\[
  \langle S_1^2\rangle_N=2(1-q)^2\qquad(N\ge3),
\]
and
\[
  \langle S_1^3\rangle_N=6(1-q)^3\qquad(N\ge4).
\]
At rank three,
\[
  \langle S_1^3\rangle_3=(1-q)^3(6+q-q^2).
\]
\end{lemma}

\begin{proof}
Let $p_r=\sum_{a=1}^N z_a^r$, so that $S_1=p_1p_{-1}$.
Consider the normalized constant-term pairing
\[
  (F,G)_{N,q}
  :=
  Z_N(q)^{-1}\frac1{N!}\CT_{\mathbf z}\,
  F(\mathbf z)\,G(\mathbf z^{-1})
  \prod_{i\ne j}(z_i/z_j;q)_\infty .
\]
Then $\langle S_1^r\rangle_N=(p_1^r,p_1^r)_{N,q}$.
In the stable range $r<N$, this finite-$N$ constant-term pairing agrees with
the Hall--Littlewood/Macdonald power-sum scalar product
for the constant-term normalization used here~\cite{Macdonald-book}
\[
  (p_\lambda,p_\mu)_q
  =
  \delta_{\lambda\mu}\,
  z_\lambda\prod_i(1-q^{\lambda_i}).
\]
Taking $\lambda=\mu=(1^r)$ gives
\[
  (p_1^r,p_1^r)_q=r!(1-q)^r.
\]
This proves
\[
  \langle S_1\rangle_N=1-q\quad(N\ge2),\qquad
  \langle S_1^2\rangle_N=2(1-q)^2\quad(N\ge3),
\]
and
\[
  \langle S_1^3\rangle_N=6(1-q)^3\quad(N\ge4).
\]

It remains to compute the first non-stable case $\langle S_1^3\rangle_3$.
A direct type-$A_2$ constant-term calculation gives
\[
\begin{array}{c|c|c}
\alpha & \text{multiplicity in } S_1^3 &
\langle z^\alpha\rangle_3 \\ \hline
(0,0,0) & 93 & 1\\[3pt]
(1,0,-1) & 360 & -\dfrac{2+q}{6}\\[6pt]
(2,0,-2) & 90 & -\dfrac{1+2q^2}{6}\\[6pt]
(2,-1,-1) & 72 & \dfrac{1+q+q^2}{3}\\[6pt]
(1,1,-2) & 72 & \dfrac{1+q+q^2}{3}\\[6pt]
(3,0,-3) & 6 & \dfrac{q+2q^2+2q^4+q^5}{6}\\[6pt]
(3,-1,-2) & 18 & -\dfrac{q+q^2+q^4}{6}\\[6pt]
(2,1,-3) & 18 & -\dfrac{q+q^2+q^4}{6}
\end{array}
\]
where $\alpha$ denotes the sorted exponent vector and the expectation is
normalized by $Z_3(q)$.  Multiplying by the displayed multiplicities and
summing yields
\[
  \langle S_1^3\rangle_3
  =
  6-17q+14q^2-4q^4+q^5
  =
  (1-q)^3(6+q-q^2).
\]
\end{proof}

\section{Onset derivations}
\label{app:onset-proofs}

This appendix collects the derivations of the delayed-onset families deferred from Section~\ref{sec:zero-intervals}.

Write
\[
  \omega(k)=\frac{k(3k-1)}{2}.
\]
The divisor form of the pentagonal-pair condition is
\begin{equation}
  \omega(r)-\omega(s)
  =\frac{(r-s)\bigl(3(r+s)-1\bigr)}{2}.
  \label{eq:pent-factor}
\end{equation}
Setting $u=r-s$ and $v=3(r+s)-1$, a valid factorization must satisfy
\begin{equation}
  v\equiv -1\pmod{3},
  \qquad
  u\equiv\frac{v+1}{3}\pmod{2}.
  \label{eq:divisor-cond}
\end{equation}

\begin{proof}[Proof of Proposition~\ref{prop:onset-divisor}]
Given $r,s\in\mathbb Z$ with $\omega(r)-\omega(s)=m$, set $u=r-s$ and
$v=3(r+s)-1$.  Then $uv=2m$ by~\eqref{eq:pent-factor}, the congruence
$v\equiv-1\pmod3$ holds identically, and $u-(v+1)/3=u-(r+s)=-2s$ is even, so
the factorization is admissible.  Conversely, an admissible factorization
determines
\[
  r=\frac{(v+1)/3+u}{2},
  \qquad
  s=\frac{(v+1)/3-u}{2},
\]
which are integers by the parity condition and satisfy
$\omega(r)-\omega(s)=m$.  The map $\omega$ is injective on $\mathbb Z$:
$\omega(r)=\omega(s)$ forces $(r-s)(3(r+s)-1)=0$, and $3(r+s)=1$ is
impossible.  Hence $(u,v)\leftrightarrow(r,s)\leftrightarrow(b,a)
=(\omega(s),\omega(r))$ is a bijection between admissible factorizations and
active pairs $b\in B_m$, and distinct admissible factorizations give distinct
$b$, hence distinct candidate energies $j=6J_R+6b$; in particular the
minimizer of~\eqref{eq:onset-min} is unique.  For the energy formula,
$r^2+s^2=\tfrac12\bigl[((v+1)/3)^2+u^2\bigr]$ gives
\[
  a+b=\frac{3(r^2+s^2)-(r+s)}{2}
  =\frac{(v+1)(v-1)}{12}+\frac{3u^2}{4}
  =\frac{9u^2+v^2-1}{12},
\]
hence $j=3(a+b)=(9u^2+v^2-1)/4$, which is~\eqref{eq:onset-energy}; minimizing
$j$ at fixed $m$ is the same as minimizing $b$, giving~\eqref{eq:onset-min}.
At $j=j^*(J_R)$ only the minimizing pair contributes, with $n=0$
in~\eqref{eq:formula}, so the onset coefficient is
\[
  E(a_*)E(b_*)\,\pfn(0)^3
  =(-1)^{r_*}(-1)^{s_*}
  =(-1)^{(v_*+1)/3},
\]
which is~\eqref{eq:onset-sign}.  Finally, $(u,v)=(-2m,-1)$ is admissible
($(v+1)/3=0$ and $u$ is even), corresponding to $(r,s)=(-m,m)$ with candidate
energy $9m^2$.
\end{proof}

The two delayed-onset families used in the main text are
\begin{equation}
  j^*(2^{n-1})=2^{2n}+2=4J_R^2+2,
  \qquad n\ge2,
  \label{eq:pow2}
\end{equation}
and, for $m=2J_R=2^\beta3^\alpha$ with $\alpha\ge1$,
\begin{equation}
  j^*(J_R)
  =3m+6\,\omega\!\left(
  \frac{1+(-1)^\beta(2^{\beta+1}-3^{\alpha+1})}{6}
  \right).
  \label{eq:twoadic}
\end{equation}

\paragraph{Power-of-two family~\eqref{eq:pow2}.}
The goal is to find the smallest onset $j^*$ when $m=2J_R=2^n$.
Every valid factorization of $uv=4J_R=2^{n+1}$ satisfying~\eqref{eq:divisor-cond} gives a candidate onset value $j=(v^2+9u^2-1)/4$; the minimum over all such factorizations is $j^*(J_R)$.

The trivial pair $(u,v)=(-2m,-1)$ always exists and gives $j=9m^2=36J_R^2$.  We show that the only other admissible factorization is the consecutive-index pair $(u,v)=((-1)^n,(-1)^n 2^{n+1})$ with $|u|=1$, giving $j^*=2^{2n}+2$.

To see this, write $v=\pm 2^b$ (since $v$ must divide $2^{n+1}$).  The congruence $v\equiv -1\pmod{3}$ fixes the sign of $v$ in terms of~$b$.  For every admissible $b\ge1$, the integer $(v+1)/3$ is odd.  The parity condition in~\eqref{eq:divisor-cond} therefore forces $u$ to be odd.  Since $uv=2^{n+1}$, this gives $|u|=1$, i.e.\ the consecutive-index pair.  The remaining case $v=-1$ gives the trivial pair.  Comparing the two values of $(v^2+9u^2-1)/4$: since $2^{2n}+2<36\cdot2^{2n-2}$ for $n\ge2$, the consecutive-index pair gives the smaller onset.

\paragraph{Two-adic family~\eqref{eq:twoadic}.}
Here $m=2J_R=2^\beta3^\alpha$ with $\alpha\ge1$.  The strategy is the same: enumerate all admissible factorizations of $uv=4J_R$ and find the minimum onset.

Since $v\equiv -1\pmod{3}$, the divisor $v$ cannot contain a factor of~$3$.  Hence $v=\pm 2^b$ and $u=\pm 2^{\beta+1-b}3^\alpha$.  The case $b=0$ gives the trivial factorization $v=-1$, with onset $j=9m^2$.

For $b\ge1$, the congruence $v\equiv -1\pmod{3}$ fixes the sign of~$v$, and $(v+1)/3$ is odd.  The parity condition then forces $u$ to be odd, so $b=\beta+1$ (absorbing all powers of~$2$ into~$v$).  This leaves exactly one non-trivial pair: $(u,v)=((-1)^\beta 3^\alpha,\,(-1)^\beta 2^{\beta+1})$.  One verifies directly that $v\equiv-1\pmod{3}$ holds (since $2^{\beta+1}\equiv(-1)^{\beta+1}\pmod{3}$) and that the parity condition is satisfied.  The onset from this pair is~\eqref{eq:twoadic}, which is smaller than $9m^2$.

\paragraph{Support above onset.}
The following properties hold for the multigraviton coefficient $d_\infty^{\rm eqQ}(j,J_R)$ above the onset energy $j\ge j^*(J_R)$:
\begin{enumerate}
\item[(i)] every active pentagonal pair contributes in every allowed energy grade above its onset;
\item[(ii)] each individual pentagonal-pair contribution is nonzero;
\item[(iii)] exact cancellation between two pentagonal-pair contributions occurs only at $(j,J_R)=(9,\tfrac12)$;
\item[(iv)] exact cancellation among three pentagonal-pair contributions is impossible;
\item[(v)] no cancellation among four or more pentagonal-pair contributions occurs at any grade with $j\le9000$.
\end{enumerate}

\begin{proof}[Proofs of (i)--(v)]
For~(i), fix an active pair $b\in B_m$.  Its contribution at $j=6J_R+6b+6n$ is $E(m+b)E(b)\,\pfn(n)^3$.  Since $E(m+b)E(b)=\pm1$ and $\pfn(n)>0$ for all $n\ge0$, this individual pair contributes nontrivially in every allowed energy grade above its onset.  Separately, the identity $\omega(-m)-\omega(m)=m$ shows that $B_m$ is nonempty for every $m\ge0$.

Statement~(ii) holds because $\pfn(n)>0$ for all $n\ge0$.

For~(iii), we must show that two-term cancellation among pentagonal pairs is extremely rare.  Suppose exactly two pairs $(b_1,a_1)$ and $(b_2,a_2)$ with $b_1<b_2$ contribute at some~$j$.  Write $n_i=(j-6J_R-6b_i)/6$.  Since $b_1<b_2$ we have $n_1>n_2\ge0$.  Cancellation requires
\[
  E(a_1)E(b_1)\,\pfn(n_1)^3 +E(a_2)E(b_2)\,\pfn(n_2)^3=0,
\]
so the two terms must have opposite signs and equal absolute values: $\pfn(n_1)^3=\pfn(n_2)^3$.  Since $\pfn(n)$ is strictly increasing for $n\ge1$~\cite{Andrews}, equality forces $\{n_1,n_2\}=\{1,0\}$ (both giving $\pfn=1$).  This means $b_2-b_1=1$ (consecutive pentagonal numbers) and $n_2=0$ (so $j=6J_R+6b_2$).

The only consecutive integers that are both generalized pentagonal are $(0,1)$ and $(1,2)$.  Taking $b_1=0,\,b_2=1$ requires $m=2J_R$ and $m+1$ both pentagonal with opposite sign products $E(m)E(0)=-E(m{+}1)E(1)$; this holds only for $m=1$.  Taking $b_1=1,\,b_2=2$ requires $m+1$ and $m+2$ both pentagonal, giving only $m=0$; but $m=0$ has same-sign products.  Hence $J_R=\tfrac12$ is the unique case, and two-term cancellation occurs only at $j=9$.

For~(iv), a cancellation among exactly three nonzero terms would have one term
on one side and two on the other, hence would imply
$X^3=Y^3+Z^3$ for positive integers $X,Y,Z$ after setting
$X=\pfn(n_i)$ for the term of one sign and $Y,Z$ for the two terms of the
opposite sign.  Euler's $n=3$ case of Fermat's theorem~\cite{HW} rules this out.  The
only partition-value collision at the bottom, $\pfn(0)=\pfn(1)=1$, would give
$X^3=2$, which is also impossible in integers.

Statement~(v) is verified computationally by an exhaustive scan of all relevant grades.  Four or more
pairs can contribute at $(j,J_R)$ only when $j\ge j_4(m)=3m+6b_4$, with $b_4$
the fourth-smallest element of $B_m$; since $j_4(m)\ge3m$, only $m\le3000$ is
relevant for $j\le9000$, and the census of admissible factorizations over
this range shows that the smallest fourth-pair onset in the entire plane is
$j_4(5)=225$, at $(j,J_R)=(225,\tfrac52)$, where $B_5=\{0,2,7,35\}$ and the
contributions are $\pfn(35)^3-\pfn(33)^3-\pfn(28)^3+\pfn(0)^3\ne0$.  All
$108{,}102$ grades with four or more contributing pairs and $j\le9000$ have
been checked to be nonzero.
\end{proof}


\section{Additional data}
\label{app:omitted-data}

Additional computations have been performed but are not needed for the main claim about coefficients with $d_\infty^{\rm eqQ}=0$ and are therefore not included here.  These include equal-fugacity comparisons across $N=2$--$5$, spin-summed $p=1$ totals with saddle-point asymptotics, the complete $j=90$ spin-resolved slice, and charge-difference profile tables.  Extended data and reproducible code beyond the ancillary files are available from the author upon request.

\end{document}